\documentclass[11pt]{article}

\usepackage[a4paper,margin=1in]{geometry}
\usepackage{amsmath,amssymb,amsfonts,amsthm,bm,mathrsfs}
\usepackage{graphicx}
\usepackage{xcolor}
\usepackage{booktabs}
\usepackage{tabularx}
\usepackage{multirow}
\usepackage{tikz}
\usepackage{enumitem}
\usepackage{array}
\usepackage{float}
\usepackage{placeins}
\usetikzlibrary{arrows.meta,positioning,fit,backgrounds,calc}
\newcommand{\E}{\mathbb E}
\newcommand{\R}{\mathbb R}

\newcommand{\sg}{\operatorname{sg}}
\newcommand{\dd}{\mathrm d}
\newcommand{\DCBSDE}{\textsc{DC--BSDE}}
\newcommand{\one}{\bm{1}}

\newtheorem{proposition}{Proposition}[section]
\newtheorem{lemma}[proposition]{Lemma}
\newtheorem{corollary}[proposition]{Corollary}
\theoremstyle{remark}
\newtheorem{remark}{Remark}[section]

\usepackage{algorithm}
\usepackage{algpseudocode}

\floatname{algorithm}{Algorithm}
\algrenewcommand\algorithmicrequire{\textbf{Input:}}
\algrenewcommand\algorithmicensure{\textbf{Output:}}
\algrenewcommand\algorithmiccomment[1]{\hfill$\triangleright$ #1}
\algrenewcommand\algorithmicif{\textbf{if}}
\algrenewcommand\algorithmicthen{\textbf{then}}
\algrenewcommand\algorithmicelse{\textbf{else}}
\algrenewcommand\algorithmicfor{\textbf{for}}
\algrenewcommand\algorithmicdo{\textbf{do}}
\algrenewcommand\algorithmicend{\textbf{end}}

\usepackage{hyperref}
\hypersetup{
  hidelinks,
  pdftitle={Deep-Control BSDE: Layerwise Brownian-Weighted Regression for High-Dimensional Semilinear PDEs},
  pdfauthor={Mingcan Wang and Xiangjun Wang},
  pdfkeywords={high-dimensional PDEs, BSDEs, Brownian-weighted regression, scientific machine learning}
}

\begin{document}

\title{Deep-Control BSDE: Layerwise Brownian-Weighted Regression for High-Dimensional Semilinear PDEs}
\author{Mingcan Wang \quad and \quad Xiangjun Wang\\[0.6em]
\small School of Mathematics and Statistics, Huazhong University of Science and Technology\\
\small Wuhan 430074, China\\
\small Corresponding author: \href{mailto:xjwang@hust.edu.cn}{xjwang@hust.edu.cn}}
\date{}

\maketitle

\begin{abstract}

High-dimensional semilinear parabolic partial differential equations arise in stochastic control, financial engineering, and uncertainty quantification, but classical spatial discretizations suffer from the curse of dimensionality. Motivated by Gaussian perturbation and conditional regression in denoising score matching, we propose Deep-Control BSDE (\DCBSDE{}), a layerwise control-regression method for Markovian backward stochastic differential equations. From an implicit Euler scheme, we identify the discrete ideal control $z_n^\pi=\mathcal M_nu_{n+1}^\pi$ as the conditional projection coefficient of the successor value response onto the one-step Brownian increment. At each time level, the method freezes the successor value function, approximates the resulting Brownian conditional-moment target using finitely many branches, regresses the control, and then fits the value through the implicit BSDE relation. A frozen linear-response baseline and antithetic pairing reduce finite-branch fluctuations while preserving the conditional target, and a within-layer correction coordinates the value and control approximations. We establish backward stability through a contraction property of the Brownian projection and derive conditional consistency when discretization, local learning, numerical, and finite-branch errors vanish jointly. Experiments across six benchmarks demonstrate that \DCBSDE{} achieves a favorable overall balance among value accuracy, control accuracy, and dynamic consistency while exhibiting generally stable performance across random seeds. These results support the use of \DCBSDE{} in applications requiring reliable joint approximation of the value and control processes.

\end{abstract}


% ============================================================================
% ============================================================================
\section{Introduction}
\label{sec:introduction}

High-dimensional semilinear parabolic partial differential equations (PDEs) arise in stochastic control, financial pricing and hedging, reaction--diffusion systems, and uncertainty quantification. Finite-difference, finite-element, and spectral methods have well-developed theories of stability, convergence, and error analysis for low-dimensional problems \cite{MortonMayers2005,Thomee2006,CanutoEtAl2006}. Their numbers of degrees of freedom, however, typically grow rapidly with the spatial dimension. This curse of dimensionality makes direct application to problems with tens or hundreds of dimensions prohibitively expensive \cite{EHanJentzen2017,HurePhamWarin2020}.

A central question in high-dimensional scientific computing is therefore how to approximate both the PDE solution and its diffusion-directional gradient without constructing an explicit high-dimensional spatial grid. Physics-informed neural networks (PINNs)\cite{RaissiEtAl2019,ZhangEtAl2026} and the Deep Galerkin Method (DGM) \cite{SirignanoSpiliopoulos2018} represent the PDE solution by a space--time neural network and minimize residuals of the governing equation, terminal condition, and boundary conditions using automatic differentiation. These approaches produce a continuous space--time approximation directly. In high-dimensional, stiff, or low-regularity regimes, however, the automatic differentiation of high-order derivatives, insufficient coverage of the state space by collocation points, and scale imbalance among different loss terms may substantially complicate optimization \cite{WangEtAl2021,KrishnapriyanEtAl2021}.

Stochastic analysis provides an alternative mesh-free route for high-dimensional semilinear PDEs. Pardoux and Peng \cite{PardouxPeng1992} established the connection between nonlinear backward stochastic differential equations (BSDEs) and semilinear parabolic PDEs. In the Markovian setting and under suitable regularity assumptions, the nonlinear Feynman--Kac formula turns the computation of the value function and its diffusion-directional gradient into the numerical solution of a BSDE. The Deep BSDE method propagates a time-discretized BSDE in forward form, treats the initial value and the control approximations at all time levels as trainable objects, and jointly optimizes them through the global terminal loss $\mathcal L_{\mathrm{term}}(\theta)=\mathbb E\!\left[\left|Y_N^\theta-g(X_N)\right|^2\right]$ \cite{EHanJentzen2017,HanJentzenE2018}.

In contrast to global pathwise optimization, DBDP1 and DBDP2 proceed backward from the terminal condition according to a dynamic-programming relation: DBDP1 jointly approximates the value and control functions through a one-step BSDE quadratic loss, whereas DBDP2 learns the value function and recovers the control from its spatial gradient \cite{HurePhamWarin2020}. More recently, Deep Picard Iteration (DPI) has reformulated high-dimensional PDE solution as a sequence of space--time regression problems with value and gradient labels derived from Feynman--Kac and variance-reduced Bismut--Elworthy--Li representations \cite{HanHuLongZhao2026}. Although these sampling-based methods avoid high-dimensional spatial grids, they still incur errors from time discretization, Monte Carlo sampling, function approximation, and layerwise backward propagation, with control identification remaining particularly delicate. In particular, neither DBDP1 nor DBDP2 constructs, from a frozen successor value function, a separate computable conditional-regression target for the current control that supports variance reduction and independent analysis. Developing such explicit control supervision is the main focus of this work.

Our point of departure is the gradient-learning paradigm of denoising score matching (DSM) \cite{Vincent2011,SongEtAl2021}. In Gaussian corruption models, DSM uses the conditional score of a known transition kernel to construct random labels and then recovers the score of the unknown marginal density by conditional regression. The control of a Markovian BSDE, $Z=\sigma^\top\nabla_xu$, is likewise a gradient-related quantity. This analogy suggests that a one-step Brownian perturbation may be used to construct explicit supervision for $Z$. It is important, however, to distinguish the two objects. DSM learns the logarithmic gradient of a probability density, $\nabla_x\log p$, whereas the present method learns the diffusion-directional gradient of a value function, $\sigma^\top\nabla_xu$. All control targets below are therefore derived rigorously from the chosen discrete BSDE rather than defined by analogy with score matching.

Malliavin weights and conditional-moment identities provide the corresponding stochastic-analytic basis for explicit control supervision. Previous work combines Malliavin-weight representations with least-squares regression to approximate the value and control processes of BSDEs \cite{GobetTurkedjiev2016}. The One Step Malliavin method further develops BSDE discretizations based on Malliavin derivatives together with deep-regression implementations \cite{NegyesiEtAl2024}. Takahashi et al. construct control variates for Deep BSDE solvers through a linear--nonlinear decomposition \cite{TakahashiEtAl2022}.

The main contributions of this paper are as follows.
\begin{enumerate}[label=(\roman*),leftmargin=*]

\item
\textbf{A control-first layerwise backward formulation with explicit Brownian supervision.}
Starting from an implicit Euler discretization, we identify the discrete ideal control as
$z_n^\pi=\mathcal M_nu_{n+1}^\pi$ and distinguish it from the frozen population target
$z_n^{\mathrm{tar}}=\mathcal M_nU_{n+1}^{\mathrm{snap}}$, the finite-branch label
$\widehat Z_n^M$, and the learned approximation $Z_n^\theta$. Freezing the successor value function turns the Brownian conditional moment into a fixed and analyzable supervision target for the current control. Each layer then follows the sequence ``freeze the successor value--regress the control--fit the implicit value--coordinate and store the snapshots.'' Relative to global end-to-end formulations, the layerwise backward organization localizes the optimization at each time level and makes the propagation of approximation errors through frozen snapshots explicit. This provides a distinct control-identification mechanism from the joint value--control residual of DBDP1 and the value-gradient recovery of DBDP2.

\item
\textbf{Conditional-mean-preserving variance reduction and backward error analysis.}
Frozen centering, a matched linear-response baseline, and antithetic Brownian pairing reduce finite-branch fluctuations without changing the population control target. We prove the Brownian projection contraction $h_n|\mathcal M_ne(x)|^2\leq (P_n|e|^2)(x)-|(P_ne)(x)|^2,$ and use it to establish stability of the exact discrete operators and local-residual-driven backward recursions. The resulting error decomposition separates time-discretization, local learning and numerical, and finite-branch errors, yielding conditional consistency when these error classes vanish jointly.

\item
\textbf{Joint numerical assessment of value accuracy, control accuracy, and dynamic consistency.}
We evaluate the methods using value and control errors together with one-step BSDE residuals and terminal-rollout errors. Experiments across six benchmarks and multiple paired random seeds, with comparisons against established deep BSDE solvers, are complemented by structural comparisons, component ablations, and wall-clock-matched tests. The results demonstrate a favorable overall balance among value accuracy, control accuracy, and dynamic consistency, together with generally stable performance across random seeds, while also showing that these criteria need not produce the same method ranking.

\end{enumerate}

The remainder of the paper is organized as follows. 
Section~\ref{sec:preliminaries} introduces the PDE--BSDE correspondence and derives the hierarchy of discrete ideal controls, frozen population targets, finite-branch labels, and variance-reduction mechanisms. Section~\ref{sec:method} presents the layerwise backward \DCBSDE{} algorithm and its optional structural interfaces, followed by the error decomposition, backward stability, conditional consistency, and complexity analysis. Section~\ref{sec:experiments} compares \DCBSDE{} with established deep BSDE solvers through joint value--control and dynamic diagnostics, structural comparisons, component ablations, and wall-clock-matched tests. Section~\ref{sec:conclusion} summarizes the main findings and future directions. The appendices provide the benchmark definitions, complete numerical results, and LogHeat experiments.

% ============================================================================
% ============================================================================
\section{Theoretical Foundations}
\label{sec:preliminaries}

This section starts from the Markovian BSDE representation of a semilinear PDE and develops the discrete objects underlying our control-regression method. Following the terminology commonly used for BSDE solvers, we refer to $Z$ as the BSDE control process \cite{EHanJentzen2017}. More precisely, $Z$ is the martingale integrand in the stochastic integral of the BSDE. In the regular Markovian setting, it coincides with the diffusion-weighted gradient of the value function, $\sigma^\top\nabla_xu$.

For a fixed time-discretization scheme, we distinguish four objects: the discrete ideal control $z_n^\pi$, the frozen population target $z_n^{\mathrm{tar}}$, the finite-branch random label $\widehat Z_n^M$, and the learned function approximation $Z_n$. These objects are respectively associated with time-discretization error, error in the frozen successor function, finite-branch sampling error, and function-approximation and optimization errors. We first define the discrete ideal control and the frozen population target. We then establish the equivalence between random-label regression and the corresponding population regression risk, and analyze conditional-mean-preserving variance reduction. Backward-propagation stability is deferred to Section~\ref{sec:method}, after the algorithmic operators have been introduced.
% ============================================================================

\subsection{Semilinear PDEs, the BSDE Control Process, and Baseline Assumptions}
\label{subsec:basic_assumptions}

Let $(\Omega,\mathcal F,\mathbb F,\mathbb P)$ be a complete filtered probability space satisfying the usual conditions. Let $W=(W_t)_{0\leq t\leq T}$ be an $m$-dimensional standard Brownian motion, and let $\mathbb F=(\mathcal F_t)_{0\leq t\leq T}$ be its augmented natural filtration. We use the Euclidean norm $|\cdot|$ for vectors and the Frobenius norm $\|\cdot\|_{\mathrm F}$ for matrices. The Euclidean inner product is denoted by $\langle\cdot,\cdot\rangle$, and $A^\top$ denotes the transpose of a matrix $A$.

We consider the terminal-value problem
\begin{equation}
\left\{
\begin{aligned}
    &\partial_tu(t,x)+\mathcal Lu(t,x) +f\!\left(t,x,u(t,x),\sigma(t,x)^\top\nabla_xu(t,x)\right)=0,
      &(t,x)\in[0,T)\times\R^d,\\
    &u(T,x)=g(x), \qquad x\in\R^d,
\end{aligned}
\right.
\label{eq:pde}
\end{equation}
where 
\[
\mathcal L\varphi(t,x)=\mu(t,x)^\top\nabla_x\varphi(t,x)+\frac12\operatorname{tr}\!\left(\sigma(t,x)\sigma(t,x)^\top D_x^2\varphi(t,x)\right).
\]

Here, $\mu:[0,T]\times\R^d\to\R^d$ and $\sigma:[0,T]\times\R^d\to\R^{d\times m}$ are the drift and diffusion coefficients, respectively; $f:[0,T]\times\R^d\times\R\times\R^m\to\R$ is the generator; and $g:\R^d\to\R$ is the terminal function.

By the nonlinear Feynman--Kac representation, for any initial condition $(t,x)$ the corresponding forward diffusion and BSDE are
\begin{align}
    X_s^{t,x}&=x+\int_t^s\mu(r,X_r^{t,x})\,\dd r +\int_t^s\sigma(r,X_r^{t,x})\,\dd W_r, \qquad s\in[t,T],
\label{eq:forward_sde}\\
    Y_s^{t,x}&=g(X_T^{t,x})+\int_s^T f(r,X_r^{t,x},Y_r^{t,x},Z_r^{t,x})\,\dd r -\int_s^T (Z_r^{t,x})^\top\dd W_r.
\label{eq:bsde}
\end{align}

The following assumptions are used for the baseline Lipschitz theory developed in this paper.

\medskip 
\noindent\textit{(H1) Forward coefficients.}
The mappings $\mu$ and $\sigma$ are continuous in $(t,x)$. There exists a constant $L_X>0$ such that, for all $t,s\in[0,T]$ and $x,x'\in\R^d$,
\begin{align*}
    |\mu(t,x)-\mu(t,x')|+\|\sigma(t,x)-\sigma(t,x')\|_{\mathrm F}&\leq L_X|x-x'|,\\
    |\mu(t,x)|+\|\sigma(t,x)\|_{\mathrm F}&\leq L_X(1+|x|),\\
    |\mu(t,x)-\mu(s,x)|+\|\sigma(t,x)-\sigma(s,x)\|_{\mathrm F}&\leq L_X(1+|x|)|t-s|^{1/2}.
\end{align*}

Consequently, for every $p\geq2$, the forward SDE admits a unique strong solution satisfying the standard moment estimate
\[
\E\!\left[\sup_{r\in[t,T]}|X_r^{t,x}|^p\right] \leq C_p(1+|x|^p).
\]

\medskip 
\noindent\textit{(H2) Generator and terminal data.}
The mapping $f$ is continuous in all variables. There exist nonnegative constants $L_x,L_y,L_z,L_t$ such that
\[
|f(t,x,y,z)-f(t,x',y',z')| \leq L_x|x-x'|+L_y|y-y'|+L_z|z-z'|,
\]
and
\[
|f(t,x,y,z)-f(s,x,y,z)| \leq L_t(1+|x|+|y|+|z|)|t-s|^{1/2}.
\]
In addition, $\int_0^T|f(t,0,0,0)|^2\,\dd t<\infty$, and $g$ is globally Lipschitz continuous.

Under (H1)--(H2), both the forward SDE and the Lipschitz BSDE are well posed. If PDE~\eqref{eq:pde} admits a sufficiently smooth classical solution, the nonlinear Feynman--Kac formula gives
\begin{equation}
    Y_s^{t,x}=u(s,X_s^{t,x}),
    \qquad
    Z_s^{t,x}=\sigma(s,X_s^{t,x})^\top \nabla_xu(s,X_s^{t,x}).
\label{eq:feynman_kac}
\end{equation}
Under weaker regularity, $u$ may be interpreted as a viscosity solution. We invoke the gradient representation in~\eqref{eq:feynman_kac} only when the corresponding Sobolev regularity and probabilistic representation justify it. In arguments that do not require this representation, $Z$ is always understood as the martingale integrand of the BSDE.

\begin{remark}[Theoretical assumptions and extended benchmarks]
\label{rem:theory_vs_benchmarks}
Assumptions (H1)--(H2) are imposed to establish the baseline Lipschitz theory. Benchmarks with quadratic-gradient generators, nonsmooth terminal data, only local growth control, or degenerate boundary structures are included as empirical robustness extensions beyond this baseline setting.
\end{remark}

% ============================================================================
\subsection{Implicit Euler Discretization and Markovian Recursion}
\label{subsec:discrete_bsde}

Let
$\pi=\{0=t_0<t_1<\cdots<t_N=T\}$ be a time grid. Set $h_n=t_{n+1}-t_n$, $|\pi|=\max_{0\leq n<N}h_n$, and $\Delta W_n=W_{t_{n+1}}-W_{t_n}$, and write $\E_n[\,\cdot\,]:=\E[\,\cdot\mid\mathcal F_{t_n}]$. The Brownian increment $\Delta W_n$ is independent of $\mathcal F_{t_n}$ and follows $\mathcal N(0,h_nI_m)$.

To distinguish the continuous-time state from its time-discrete approximation, we denote the Euler--Maruyama state by $\bar X_n$:
\begin{equation}
    \bar X_{n+1}=\Phi_n(\bar X_n,\Delta W_n):=\bar X_n+\mu(t_n,\bar X_n)h_n+\sigma(t_n,\bar X_n)\Delta W_n.
\label{eq:euler}
\end{equation}
The term ``discrete state'' refers only to the temporal grid. The random variable $\bar X_n$ remains continuously valued in $\R^d$.

We take the implicit Euler BSDE with the generator evaluated at the current time as the reference discretization \cite{Zhang2004,GobetLemorWarin2005}:
\begin{subequations}
\label{eq:discrete_scheme}
\begin{align}
    Y_N^\pi&=g(\bar X_N),
\label{eq:discrete_terminal}\\
    Z_n^\pi&=\frac1{h_n}\E_n\!\left[Y_{n+1}^\pi\Delta W_n\right],
\label{eq:discrete_z}\\
    Y_n^\pi&=\E_n[Y_{n+1}^\pi]+h_nf(t_n,\bar X_n,Y_n^\pi,Z_n^\pi).
\label{eq:discrete_y}
\end{align}
\end{subequations}
We further assume
\begin{equation}
    |\pi|L_y<1.
\label{eq:implicit_solvability}
\end{equation}
Then, for fixed $(t_n,x,z)$ and any $a\in\R$, the mapping $y\mapsto a+h_nf(t_n,x,y,z)$ is a contraction, and the current-layer implicit scalar equation in~\eqref{eq:discrete_y} has a unique solution. Condition~\eqref{eq:implicit_solvability} is sufficient for a Picard iteration, but it is not necessary in general.

Fix the current state $x$, let $w\sim\mathcal N(0,h_nI_m)$, and let $\E_w$ denote expectation only with respect to the next Brownian increment. For a suitably integrable scalar function $\varphi:\R^d\to\R$, define the one-step transition operator $P_n$ and the Brownian-weighted operator $\mathcal M_n$ by
\begin{align*}
    (P_n\varphi)(x):=\E_w\!\left[\varphi(\Phi_n(x,w))\right],\quad
    (\mathcal M_n\varphi)(x):=\E_w\!\left[\varphi(\Phi_n(x,w))\frac{w}{h_n}\right].
\end{align*}
Along the Euler state chain, these definitions are equivalently written as
\begin{align*}
    (P_n\varphi)(\bar X_n)=\E[\varphi(\bar X_{n+1})\mid\mathcal F_{t_n}],\quad
    (\mathcal M_n\varphi)(\bar X_n)=\frac1{h_n}\E[\varphi(\bar X_{n+1})\Delta W_n\mid\mathcal F_{t_n}].
\end{align*}
Thus, $P_n$ maps a scalar function at the next time level to a scalar function at the current level, whereas $\mathcal M_n$ maps it to an $\R^m$-valued function at the current level.

In the Markovian setting, write $Y_n^\pi=u_n^\pi(\bar X_n)$ and $Z_n^\pi=z_n^\pi(\bar X_n)$. Then
\begin{subequations}
\label{eq:markov_recursion}
\begin{align}
    z_n^\pi(x)&=(\mathcal M_nu_{n+1}^\pi)(x),
\label{eq:markov_z}\\
    u_n^\pi(x)&=(P_nu_{n+1}^\pi)(x)+h_nf(t_n,x,u_n^\pi(x),z_n^\pi(x)).
\label{eq:value_operator}
\end{align}
\end{subequations}
Equation~\eqref{eq:markov_z} defines the discrete ideal control associated with the chosen time-discretization scheme.

% ============================================================================
\subsection{Conditional Projection, Gaussian Integration by Parts, and Gradient Interpretation}
\label{subsec:projection_stein}

Equation~\eqref{eq:discrete_z} can also be derived from a conditional least-squares projection. More precisely, $Z_n^\pi$ is the unique optimal coefficient in
\[
Z_n^\pi=\operatorname*{argmin}_{\zeta\in L^2(\mathcal F_{t_n};\R^m)}\E\!\left[\left|Y_{n+1}^\pi-\E_n[Y_{n+1}^\pi]-\zeta^\top\Delta W_n\right|^2\right].
\]
Define the projection residual $R_{n+1}^{\mathrm{proj},\pi}:=Y_{n+1}^\pi-\E_n[Y_{n+1}^\pi]-(Z_n^\pi)^\top\Delta W_n.$
It satisfies the orthogonality conditions $\E_n[R_{n+1}^{\mathrm{proj},\pi}]=0, \E_n[R_{n+1}^{\mathrm{proj},\pi}\Delta W_n]=0.$
Therefore, $Z_n^\pi$ is the projection coefficient of the centered successor value response onto the conditional linear space spanned by the components of $\Delta W_n$. It extracts the first-order Brownian component of the one-step response, without requiring $Y_{n+1}^\pi$ to be exactly linear in $\Delta W_n$. Equation~\eqref{eq:discrete_z} is an exact identity for the selected discrete scheme. Relative to the continuous-time control $z(t_n,x)=\sigma(t_n,x)^\top\nabla_xu(t_n,x)$, however, it still contains time-discretization error.

Let $q_{h_n}$ denote the density of $\mathcal N(0,h_nI_m)$ and define the one-step Brownian--Stein weight $H_n(w):=\frac{w}{h_n}=-\nabla_w\log q_{h_n}(w).$
If $\Phi_n(x,\cdot)$ and $\varphi$ are weakly differentiable and the relevant integrability conditions hold, Gaussian integration by parts gives
\[
\E_w\!\left[\varphi(\Phi_n(x,w))\frac{w}{h_n}\right]
=\E_w\!\left[D_w\Phi_n(x,w)^\top\nabla\varphi(\Phi_n(x,w))\right].
\]
For the Euler--Maruyama map~\eqref{eq:euler}, $D_w\Phi_n(x,w)=\sigma(t_n,x)$, and hence
\begin{equation}
    (\mathcal M_n\varphi)(x) =\E_w\!\left[ \sigma(t_n,x)^\top \nabla\varphi(\Phi_n(x,w)) \right].
\label{eq:euler_stein_identity}
\end{equation}
Thus, the Brownian-weighted conditional moment equals the conditional average of the diffusion-directional gradient of the successor function \cite{Stein1981}.

\begin{remark}[Relation to denoising score matching]
\label{rem:dsm_connection}
Let the latent variable be $\xi$ and the corrupted variable be $\widetilde\xi=\xi+\sqrt h\,\varepsilon$, where $\varepsilon\sim\mathcal N(0,I)$. DSM constructs random supervision from the conditional score $-\varepsilon/\sqrt h$ of the corruption kernel and recovers the logarithmic gradient of the marginal density by conditional regression. The weight $w/h_n=-\nabla_w\log q_{h_n}(w)$ used here is the corresponding Gaussian Stein weight. The two settings share Gaussian conditional regression and integration-by-parts identities, but they weight different responses and target different quantities. DSM learns $\nabla\log p$, whereas the present method weights a successor value response and identifies a discrete BSDE control. DSM therefore provides motivation for the supervision design only; the target definition, conditional-mean property, and error analysis in this paper are all determined by the discrete BSDE. A strict intersection of the two constructions for the log-heat equation is discussed in Appendix~\ref{app:score_case}.
\end{remark}

\begin{remark}[Difference from continuous-time Malliavin weights]
\label{rem:discrete_vs_bismut_weight}
The factor $w/h_n$ is an integration-by-parts weight for the current one-step Gaussian transition. We refer to it as a one-step Brownian--Malliavin or Brownian--Stein weight. In general, it is not the continuous-time Bismut--Elworthy--Li weight, which typically depends on a stochastic-flow Jacobian, information along an entire path segment, and corresponding nondegeneracy conditions \cite{ElworthyLi1994,MaZhang2002,GobetTurkedjiev2016}.
\end{remark}

% ============================================================================
\subsection{Frozen Population Targets and Population Regression Risk}
\label{subsec:frozen_targets}

Before entering time level $n$, the trained successor function $U_{n+1}^{\mathrm{snap}}$ is frozen strictly. Let $\mathcal G_n^{\mathrm{pre}}$ denote the algorithmic information available before the current outer states and Brownian branches are generated, including the training history, frozen functions, and predetermined variance-control terms.

Fix a $\mathcal G_n^{\mathrm{pre}}$-measurable state $x$ and independently draw $w\sim\mathcal N(0,h_nI_m)$. Define the one-step future response by
\begin{equation}
    \Psi_n(x,w) :=U_{n+1}^{\mathrm{snap}}(\Phi_n(x,w)).
\label{eq:future_response}
\end{equation}
The current-layer population target induced by the frozen function is
\begin{equation}
    z_n^{\mathrm{tar}}(x):=(\mathcal M_nU_{n+1}^{\mathrm{snap}})(x) =\E_w\!\left[\Psi_n(x,w)\frac{w}{h_n}\right].
\label{eq:frozen_z_target}
\end{equation}
The corresponding raw one-branch label is
\begin{equation}
    T_n^{\mathrm{raw}}(x;w):=\Psi_n(x,w)H_n(w)=\Psi_n(x,w)\frac{w}{h_n}.
\label{eq:raw_brownian_label}
\end{equation}
By conditional independence of the current increment and $\mathcal G_n^{\mathrm{pre}}$,
\begin{equation}
    \E\!\left[T_n^{\mathrm{raw}}(x;w) \mid\mathcal G_n^{\mathrm{pre}} \right] =z_n^{\mathrm{tar}}(x).
\label{eq:raw_label_mean}
\end{equation}
Thus, the raw label is conditionally unbiased for the frozen population target, but it is not a pathwise realization of the discrete ideal control $z_n^\pi$. The difference between the two targets satisfies the exact identity
\begin{equation}
    z_n^{\mathrm{tar}}(x)-z_n^\pi(x) =\Bigl[\mathcal M_n \bigl(U_{n+1}^{\mathrm{snap}}-u_{n+1}^\pi\bigr) \Bigr](x).
\label{eq:frozen_target_error}
\end{equation}

If $u_{n+1}^\pi$ were known, the ideal control risk under a prescribed state design measure $\nu_n^{\mathrm{des}}$ would be
\[
\mathcal L_{Z,n}^{\mathrm{ideal}}(F) :=\frac12\int_{\R^d} |F(x)-z_n^\pi(x)|^2\,\nu_n^{\mathrm{des}}(\dd x).
\]
Choosing $\nu_n^{\mathrm{des}}=\operatorname{Law}(\bar X_n)$ yields the risk under the path distribution. Since $u_{n+1}^\pi$ is unknown, the algorithm learns the frozen population target $z_n^{\mathrm{tar}}$ instead.

\begin{proposition}[Random labels and the population regression target]
\label{prop:population_regression_target}
Suppose that, conditional on $\mathcal G_n^{\mathrm{pre}}$, the distribution of $X_n^{\mathrm{des}}$ is the predetermined design measure $\nu_n^{\mathrm{des}}$, and that the random label $\widehat T_n$ satisfies
\begin{equation}
    \E\!\left[ \widehat T_n \mid\mathcal G_n^{\mathrm{pre}},X_n^{\mathrm{des}} \right] =z_n^{\mathrm{tar}}(X_n^{\mathrm{des}}),
\label{eq:generic_label_conditional_mean}
\end{equation}
with a finite second moment. Then, for any $F\in L^2(\nu_n^{\mathrm{des}};\R^m)$ fixed before the current design states and labels are generated,
\begin{equation}
\begin{aligned}
    &\E\!\left[ |F(X_n^{\mathrm{des}})-\widehat T_n|^2 \mid\mathcal G_n^{\mathrm{pre}} \right]= \|F-z_n^{\mathrm{tar}}\|_{L^2(\nu_n^{\mathrm{des}})}^2+\mathcal V_n,
\end{aligned}
\label{eq:population_risk_decomposition}
\end{equation}
where $\mathcal V_n:=\E\!\left[ |\widehat T_n-z_n^{\mathrm{tar}}(X_n^{\mathrm{des}})|^2\mid\mathcal G_n^{\mathrm{pre}} \right]$ is independent of $F$. Therefore, over an unrestricted function class, the unique minimizer of the population risk is $z_n^{\mathrm{tar}}$, $\nu_n^{\mathrm{des}}$-almost everywhere. Within a prescribed function class $\mathcal H_n$, minimizing the population risk is equivalent to minimizing the $L^2(\nu_n^{\mathrm{des}})$ approximation error to $z_n^{\mathrm{tar}}$ over $\mathcal H_n$.
\end{proposition}

\noindent\emph{Proof.}
Add and subtract $z_n^{\mathrm{tar}}(X_n^{\mathrm{des}})$ on the left-hand side of~\eqref{eq:population_risk_decomposition}. By~\eqref{eq:generic_label_conditional_mean}, the conditional mean of the cross term is zero given $(\mathcal G_n^{\mathrm{pre}},X_n^{\mathrm{des}})$, which yields the stated orthogonal decomposition.\hfill$\square$

Proposition~\ref{prop:population_regression_target} closes the population-level argument from ``the label preserves the conditional mean'' to ``neural regression learns the frozen population target.''

% ============================================================================
\subsection{Label Variance and Conditional-Mean-Preserving Variance-Control Terms}
\label{subsec:variance_reduction}

Let the one-step mean of the frozen response be $\mu_n^{\mathrm{tar}}(x):=(P_nU_{n+1}^{\mathrm{snap}})(x).$
The raw label admits the decomposition 
\[
T_n^{\mathrm{raw}}(x;w) =\mu_n^{\mathrm{tar}}(x)\frac{w}{h_n} +[\Psi_n(x,w)-\mu_n^{\mathrm{tar}}(x)] \frac{w}{h_n}.
\]
The first term has conditional mean zero, but its conditional second moment is
\[
\E_w\!\left[ \left|\mu_n^{\mathrm{tar}}(x)\frac{w}{h_n}\right|^2 \right] =\frac{m|\mu_n^{\mathrm{tar}}(x)|^2}{h_n}.
\]
Hence, when $h_n$ is small, the Brownian dimension $m$ is large, and the response level is non-negligible, the raw label may contain substantial zero-mean fluctuations of order $h_n^{-1}$.

\begin{remark}[Conditionally affine responses]
\label{rem:affine_response_variance}
Suppose that, conditional on $(\mathcal G_n^{\mathrm{pre}},x)$, the response is affine:
$\Psi_n(x,w)=c_n(x)+a_n(x)^\top w$. Then
\[
\E_w\!\left[|T_n^{\mathrm{raw}}(x;w)-a_n(x)|^2\right] =\frac{m|c_n(x)|^2}{h_n}+(m+1)|a_n(x)|^2.
\]
This special case makes the two sources of variance explicit: the $h_n^{-1}$ term is caused by the constant response level, while the linear response itself also induces fluctuations that grow with the Brownian dimension.
\end{remark}

To avoid confusion with the BSDE control process $Z$, we refer to Monte Carlo control variates as \emph{variance-control terms}. Let $b_n(x)\in\R$ be a frozen scalar center and $B_n(x)\in\R^m$ a frozen linear-response coefficient. Both must be determined before the current Brownian branches are generated.

With scalar centering only, define
\[
T_n^{\mathrm{ctr}}(x;w):=[\Psi_n(x,w)-b_n(x)]\frac{w}{h_n}.
\]
When an approximately linear response is removed as well, define
\[
T_n^{\mathrm{cv}}(x;w):=B_n(x)+[\Psi_n(x,w)-b_n(x)-B_n(x)^\top w] \frac{w}{h_n}.
\]

\begin{proposition}[Conditional-mean preservation of Brownian-weighted labels]
\label{prop:conditional_unbiasedness}
Suppose that $U_{n+1}^{\mathrm{snap}}$, $b_n$, $B_n$, and the current state $x$ are $\mathcal G_n^{\mathrm{pre}}$-measurable, and that $w\sim\mathcal N(0,h_nI_m)$ is conditionally independent of $\mathcal G_n^{\mathrm{pre}}$. If the relevant expectations are finite, then
\begin{equation}
\begin{aligned}
    \E[T_n^{\mathrm{raw}}(x;w)\mid\mathcal G_n^{\mathrm{pre}}]
    &=\E[T_n^{\mathrm{ctr}}(x;w)\mid\mathcal G_n^{\mathrm{pre}}]\\
    &=\E[T_n^{\mathrm{cv}}(x;w)\mid\mathcal G_n^{\mathrm{pre}}]\\
    &=z_n^{\mathrm{tar}}(x).
\end{aligned}
\label{eq:cv_mean}
\end{equation}
\end{proposition}

\noindent\emph{Proof.}
Conditioning on $\mathcal G_n^{\mathrm{pre}}$ and using $\E_w[w]=0$ and $\E_w[ww^\top]=h_nI_m$ gives
\begin{align*}
    \E_w[T_n^{\mathrm{ctr}}] &=z_n^{\mathrm{tar}},\\
    \E_w[T_n^{\mathrm{cv}}]&=B_n+z_n^{\mathrm{tar}} -\frac1{h_n}\E_w[(B_n^\top w)w] =z_n^{\mathrm{tar}}.
\end{align*}
The statement for the raw label follows directly from its definition.\hfill$\square$

\begin{remark}[Scaling of the linear-response term]
\label{rem:baseline_scaling}
If $\lambda_n(x)$ is fixed before the current branches are generated, the conditional-mean-preserving scaled form must be written as
\[
\lambda_nB_n+[\Psi_n-b_n-\lambda_nB_n^\top w]\frac{w}{h_n}.
\]
Scaling only the linear term inside the residual while retaining the unscaled external compensation $B_n$ introduces the additional bias $(1-\lambda_n)B_n$.
\end{remark}

Gaussian integration by parts also suggests a natural linear-response baseline. Let
\[
x_{n+1}^0(x):=\Phi_n(x,0).
\]
If $U_{n+1}^{\mathrm{snap}}$ is differentiable at this point, the local linear coefficient of the future response with respect to $w$ at $w=0$ is
\begin{equation}
    B_n^{\mathrm{lin}}(x):=D_w\Phi_n(x,0)^\top \nabla U_{n+1}^{\mathrm{snap}}(x_{n+1}^0(x)).
\label{eq:local_linear_baseline}
\end{equation}
For the Euler--Maruyama transition,~\eqref{eq:local_linear_baseline} reduces to
\begin{equation}
    B_n^{\mathrm{lin}}(x) =\sigma(t_n,x)^\top \nabla U_{n+1}^{\mathrm{snap}}(x_{n+1}^0(x)).
\label{eq:local_linear_baseline_euler}
\end{equation}
At the final time level, $U_N^{\mathrm{snap}}=g$, so one may use
\begin{equation}
    B_{N-1}^{\mathrm{grad}}(x) =D_w\Phi_{N-1}(x,0)^\top \nabla g(x_N^0(x)).
\label{eq:terminal_baseline}
\end{equation}
These quantities approximate only the local linear component of the future response with respect to the current Brownian increment. In general, they are not equal to the discrete ideal control $z_n^\pi(x)$. If automatic differentiation is unstable or the terminal function is nonsmooth, one may use a prespecified weak gradient, a smoothed gradient, a future control snapshot, or a zero baseline. Whatever the choice, it must be frozen before the current branches are generated. Valid freezing guarantees preservation of the conditional mean.

% ============================================================================
\subsection{Antithetic Brownian Pairing and Finite-Branch Labels}
\label{subsec:antithetic_labels}

Using $w\overset{d}= -w$, define the odd part of the future response with respect to the current increment by
\begin{equation}
    D_n(x,w):=\frac{\Psi_n(x,w)-\Psi_n(x,-w)}2.
\label{eq:odd_response}
\end{equation}
Any scalar center independent of $w$ cancels automatically in this difference. Let the total budget of future-function evaluations be an even integer $M$, set $K=M/2$, draw independent increments $w^{(k)}\sim\mathcal N(0,h_nI_m)$, and evaluate both the positive and negative branches for each increment. Define
\begin{align}
    C_n^{(k)}(x) &:=[D_n(x,w^{(k)})-B_n(x)^\top w^{(k)}]\frac{w^{(k)}}{h_n},
\label{eq:antithetic_correction}\\
    \widehat Z_n^{M}(x)&:=B_n(x)+\frac1K\sum_{k=1}^{K}C_n^{(k)}(x).
\label{eq:antithetic_target}
\end{align}
By Gaussian symmetry,
\[
\E_w\!\left[D_n(x,w)\frac{w}{h_n}\right]=z_n^{\mathrm{tar}}(x),
\]
and therefore
\[
\E[\widehat Z_n^{M}(x)\mid\mathcal G_n^{\mathrm{pre}}] =z_n^{\mathrm{tar}}(x).
\]
If the branch pairs are conditionally independent, the finite-branch mean-square fluctuation satisfies
\begin{equation}
\begin{aligned}
&\E\!\left[
    \left|\widehat Z_n^{M}(x)-z_n^{\mathrm{tar}}(x)\right|^2 \,\middle|\,\mathcal G_n^{\mathrm{pre}}
    \right] =\frac{1}{K}\,\E\!\left[
    \left|C_n^{(1)}(x)-\bigl(z_n^{\mathrm{tar}}(x)-B_n(x)\bigr)\right|^2
    \,\middle|\,\mathcal G_n^{\mathrm{pre}}
\right].
\end{aligned}
\label{eq:antithetic_mc_variance}
\end{equation}
Exploiting Gaussian symmetry, antithetic pairing extracts the odd response by canceling the even component, which has zero population contribution under the odd Brownian weight $w/h_n$, thereby reducing finite-branch variance without changing the conditional control target \cite{Glasserman2004}.

% ============================================================================
% ============================================================================

\section{The Layerwise Backward Deep-Control BSDE Algorithm}
\label{sec:method}

Section~\ref{sec:preliminaries} established the relationship among the discrete ideal control, the frozen population control target, and the finite-branch Brownian labels. Building on these objects, we now construct the layerwise backward Deep-Control BSDE (\DCBSDE{}) algorithm, organized as a control-first recursion from the terminal condition.

% ============================================================================
\subsection{Algorithmic Components and Single-Layer Backward Dependencies}
\label{subsec:overview}

Set $U_N^{\mathrm{snap}}=g$. The superscript ``$\mathrm{snap}$'' denotes a self-contained immutable function snapshot comprising the network parameters, state preprocessing, active feature transformations, structural projections, and associated interface configuration. Upon entering level $n$, $U_{n+1}^{\mathrm{snap}}$ and any future-control snapshots required by the discretization have already been fixed and are treated as immutable. The frozen successor value defines the population control target and its finite-branch labels, from which the current control is regressed. The trained control, together with the conditional mean of the successor value, then determines the implicit value target. After within-layer coordination, the resulting pair is stored as $(U_n^{\mathrm{snap}},Z_n^{\mathrm{snap}})$. Consequently, subsequent updates to shared parameters affect only the current trainable representation and cannot alter the backward targets defined by previously completed time levels.

The method consists of three groups of components: the discrete backward backbone, generic training and within-layer coordination mechanisms, and a unified equation-informed structural screening interface. Table~\ref{tab:method_components} summarizes the roles of the three groups.

\begin{table}[!htbp]
\centering
\caption{Three component groups of \DCBSDE{}}
\label{tab:method_components}
\small
\begin{tabularx}{\textwidth}{@{}p{3.2cm}X X@{}}
\toprule
Component group & Main objects & Role in the computational pipeline \\
\midrule
Discrete backward backbone
&
Frozen successor functions, Brownian control targets, the implicit value operator, and immutable snapshots
&
Defines the population control and value targets at the current level and specifies the layerwise backward dependency structure
\\

Generic training and coordination mechanisms
&
Antithetic pairing, a shared control trunk with time-local residual heads, replay, implicit fixed-point solves, compatibility regularization, control refitting, and value polishing
&
Reduces finite-sample fluctuations, stabilizes shared training across time levels, and coordinates the independently learned $U_n$ and $Z_n$
\\
Unified equation-informed structural screening interface
&
Verified symmetry projections, automatically screened equation-informed value features, and validated linear-response baselines with automatic-differentiation or zero fallbacks
&
Screens and activates prespecified structural candidates using component-specific validity and conditioning checks, freezes the retained configuration before training, and invokes generic fallbacks when a candidate is unavailable or rejected
\\
\bottomrule
\end{tabularx}
\end{table}

\begin{figure}[!tbp]
\centering
\resizebox{\textwidth}{!}{%
\begin{tikzpicture}[
  >=Latex,
  node distance=8mm and 10mm,
  core/.style={draw,very thick,rounded corners,align=center,
               minimum height=9mm,text width=27mm,fill=blue!7},
  generic/.style={draw,rounded corners,align=center,
                  minimum height=8mm,text width=29mm,fill=green!7},
  structure/.style={draw,dashed,rounded corners,align=center,
                    minimum height=8mm,text width=31mm,fill=orange!8},
  line/.style={->,thick},
  gline/.style={->,thick},
  sline/.style={->,thick,dashed}
]

\node[core] (snap)
  {Use frozen successor snapshot\\$U_{n+1}^{\mathrm{snap}}$};

\node[core,right=of snap] (ztar)
  {Generate finite-branch\\control labels $\widehat Z_n^M$};

\node[core,right=of ztar] (zfit)
  {Control regression\\$Z_n^{\mathrm{init}}$};

\node[core,right=of zfit] (ytar)
  {Construct implicit value target\\$\widehat Y_n^{\mathrm{init}}$};

\node[core,right=of ytar] (ufit)
  {Value regression\\$U_n^{\mathrm{init}}$};

\node[core,right=of ufit] (freeze)
  {Coordinate and store snapshot\\
   $(U_n^{\mathrm{snap}},Z_n^{\mathrm{snap}})$};

\node[structure,above=of ztar] (zstruct)
  {Validated control enhancements\\
   response baseline and\\
   equivariant control projection};

\node[structure,above=of ufit] (ustruct)
  {Screened value features\\
 and verified invariant projection};

\node[generic,below=of ztar] (anti)
  {Antithetic pairing, label refreshing,\\
   and variance estimation};

\node[generic,below=of zfit] (shared)
  {Shared trunk, time-local head,\\
   and replay distillation};

\node[generic,below=of ytar] (fp)
  {Conditional-mean estimation\\
   and implicit fixed-point solve};

\node[generic,below=of ufit] (compat)
  {Value--control compatibility\\
   regularization};

\node[generic,below=of freeze] (corr)
  {$Z$ refitting, value polishing,\\
   and constant-bias calibration};

\draw[line] (snap) -- (ztar);
\draw[line] (ztar) -- (zfit);
\draw[line] (zfit) -- (ytar);
\draw[line] (ytar) -- (ufit);
\draw[line] (ufit) -- (freeze);

\draw[sline] (zstruct) -- (ztar);
\draw[sline] (ustruct) -- (ufit);

\draw[gline] (anti) -- (ztar);
\draw[gline] (shared) -- (zfit);
\draw[gline] (fp) -- (ytar);
\draw[gline] (compat) -- (ufit);
\draw[gline] (corr) -- (freeze);

\end{tikzpicture}%
}
\caption{Within-layer computational dependencies at time level $n$}
\label{Fig1}
\end{figure}
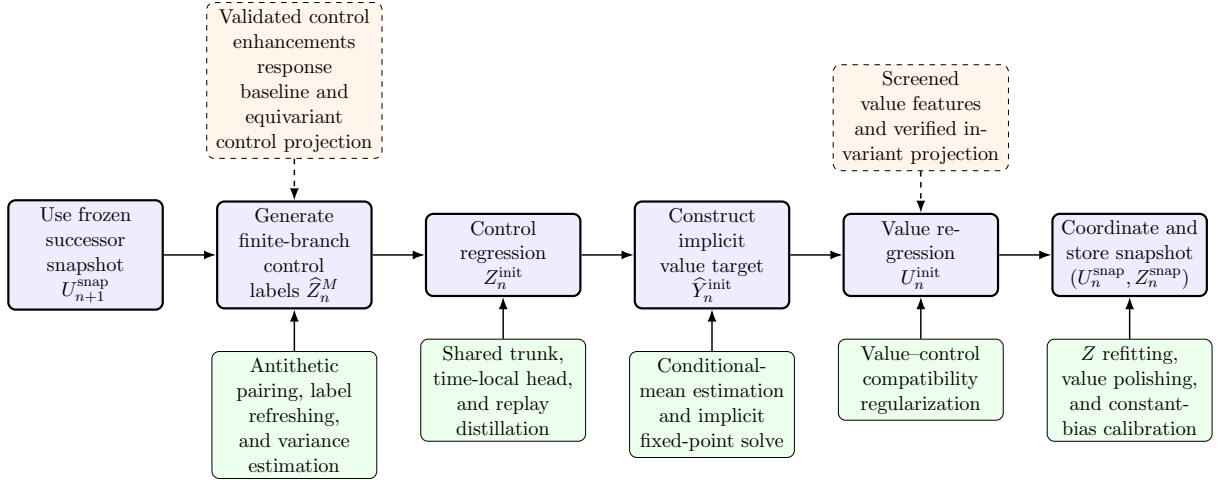

Figure~\ref{Fig1} illustrates the computational dependencies at time level $n$. The middle row represents the discrete algorithmic backbone, the lower green row contains the generic training and within-layer coordination mechanisms, and the upper dashed orange row contains candidate structural enhancements. The remainder of this section follows the actual order of computation shown in the figure.

% ============================================================================
\subsection{Pretraining Setup and Interface Conventions}
\label{subsec:training_setup}

The outer states determine the design measures for function regression, whereas the inner Brownian branches approximate conditional means and Brownian-weighted conditional moments. Training, replay, structural verification, bias calibration, validation, and testing data serve distinct purposes. The control labels, initial value targets, value-polishing targets, and bias-calibration targets use mutually independent inner random streams so that the Monte Carlo fluctuations of the different stages can be assessed separately.

Let $\nu_n^{\mathrm{tr}}$, $\nu_n^{\mathrm{rep}}$, $\nu_n^{\mathrm{bias}}$, and $\nu_n^{\mathrm{val}}$ denote the training, replay, bias-calibration, and validation state distributions, respectively. The default training design is $\operatorname{Law}(\bar X_n)$. When stronger identification is needed near the deterministic initial condition or in low-probability regions, locally perturbed states may be mixed with on-path states. The mixture proportion, perturbation scale, and state preprocessing are fixed before training. In particular, when $\bar X_0=x_0$ is deterministic and the compatibility term $\sigma^\top\nabla U_0$ is used at level zero, local perturbations provide the neighborhood coverage required to compute spatial gradients. If only $U_0(x_0)$ and $Z_0(x_0)$ are required, one may simply take $\nu_0^{\mathrm{tr}}=\delta_{x_0}$.

Before backward training, the unified equation-informed structural screening interface evaluates four optional components: a value projection $\mathcal P_{U,n}$, a control projection $\mathcal P_{Z,n}^{Q}$, value features $\bm\varphi_n$, and an explicit response baseline $B_n^{\mathrm{prob}}$. Candidates are retained only after the corresponding validity or conditioning checks. When the corresponding structure is unavailable, we set
\[
    \mathcal P_{U,n}=I,\qquad
    \mathcal P_{Z,n}^{Q}=I,\qquad
    \bm\varphi_n=\bm\varphi_n^{\mathrm{gen}}\ \text{or}\ \varnothing,
    \qquad B_n^{\mathrm{prob}}=\varnothing.
\]
The retained configuration is then fixed for all training stages. These components affect only the finite-sample representation or label variance; they do not alter the population control target, the implicit value operator, or the layerwise backward order.

\medskip 
\noindent\textit{Optional symmetry projections.}
Let $\mathscr R_cx:=2c-x$. For a candidate center $c\in\R^d$ and an orthogonal involution
$Q_n\in\R^{m\times m}$, suppose that
\[
\begin{aligned}
g(\mathscr R_cx) =g(x),\quad  
\Phi_n(\mathscr R_cx,Q_nw) =\mathscr R_c\Phi_n(x,w),\quad
f(t_n,\mathscr R_cx,y,Q_nz) =f(t_n,x,y,z),
\end{aligned}
\]
and that the discrete recursion has a unique solution. Then
\[
u_n^\pi(\mathscr R_cx)=u_n^\pi(x),\qquad
z_n^\pi(\mathscr R_cx)=Q_nz_n^\pi(x).
\]
The associated projections are
\[
\begin{aligned}
(\mathcal P_{U,n}F)(x) :=\frac{F(x)+F(\mathscr R_cx)}{2}, \quad
(\mathcal P_{Z,n}^{Q}F_Z)(x) :=\frac{F_Z(x)+Q_nF_Z(\mathscr R_cx)}{2}.
\end{aligned}
\]
They impose the corresponding invariance or equivariance on the value model, the direct control network, and the gradient-induced control \cite{PunyEtAl2022}.

% ============================================================================
\subsection{Brownian-Weighted Control Labels and Control Regression}
\label{subsec:z_estimation}

% ============================================================================
\subsubsection{Frozen successor value and the current control target}

At time level $n$, we first freeze $U_{n+1}^{\mathrm{snap}}$. By~\eqref{eq:frozen_z_target}, the population target for the current control is
$z_n^{\mathrm{tar}}=\mathcal M_nU_{n+1}^{\mathrm{snap}}$, which depends only on the successor snapshot and the one-step transition. Hence, $Z_n$ can be trained before $U_n$. This ordering separates the control-regression error associated with the Brownian-weighted conditional moment from the value-regression error associated with the successor conditional mean and the implicit generator relation.

% ============================================================================
\subsubsection{Frozen linear-response baselines and antithetic labels}
To stabilize the finite-branch control labels, we combine a frozen linear-response baseline with antithetic pairing. The baseline subtracts an approximation to the dominant first-order response, while antithetic pairing cancels the even component of the successor response.

Let $\sg(\cdot)$ denote the stop-gradient operator. The generic local linear-response baseline $B_n^{\mathrm{lin}}$ is defined in~\eqref{eq:local_linear_baseline}. For $0\leq n\leq N-2$, a frozen control from the adjacent future layer may alternatively provide the response scale
\begin{equation}
B_n^{\mathrm{fut}}(x):=\sg\left[Z_{n+1}^{\mathrm{snap}}\left(\Phi_n(x,0)\right)\right].
\label{eq:future_control_baseline}
\end{equation}
Here, $B_n^{\mathrm{fut}}$ approximates the neighboring response shape; it is not used as the current control. At the terminal-adjacent level $n=N-1$, this future-control baseline is unavailable, and the terminal-gradient baseline in~\eqref{eq:terminal_baseline} may be used instead. The response interface may also supply a validated analytic candidate $B_n^{\mathrm{prob}}$. Because the external compensation and the internal linear term in~\eqref{eq:antithetic_target} form a matched pair, the selected baseline changes finite-branch fluctuations without altering the population target $\mathcal M_nU_{n+1}^{\mathrm{snap}}$.

Let the total successor-evaluation budget per outer state be an even integer $M$, and set $K=M/2\geq2$. Draw independent increments
$w^{(k)}\sim\mathcal N(0,h_nI_m)$, evaluate the positive and negative branches, and construct $\widehat Z_n^M(x)$ according to~\eqref{eq:odd_response}--\eqref{eq:antithetic_target}. Since the Gaussian law is symmetric and the Brownian weight is odd, the even response component has zero population contribution; antithetic pairing cancels it within each sample pair, thereby reducing finite-branch variance without changing the conditional target $z_n^{\mathrm{tar}}$.

Define
\[
\overline C_n(x):=\frac1K\sum_{k=1}^{K}C_n^{(k)}(x).
\]
If the sample pairs are conditionally independent and identically distributed given the frozen information, an unbiased estimator of the conditional covariance trace of the label mean is
\begin{equation}
\widehat v_n^M(x):=\frac{1}{K(K-1)}\sum_{k=1}^{K}\left|C_n^{(k)}(x)-\overline C_n(x)\right|^2.
\label{eq:target_variance}
\end{equation}
This quantity monitors state-dependent label fluctuations and is later used to assess the finite-sample reliability of the Brownian labels during control refitting.

% ============================================================================
\subsubsection{Shared temporal trunk, local residual heads, and replay distillation}

The shared trunk captures time--state structure common across time levels, while the local residual heads represent level-specific deviations. This design balances sample and parameter reuse with per-level flexibility, avoiding both the limited expressiveness of a fully shared model and the loss of temporal coupling associated with independent networks. We therefore parameterize the control as
\[
\begin{aligned}
    \widetilde Z_n^{\mathrm{train}}(x) =Z_{\theta^{\mathrm{sh}}}^{\mathrm{sh}}(t_n,x)+\eta_{Z,n}\mathcal N_{\psi_n}^{Z}(x),\quad
    Z_n^{\mathrm{train}}(x) =\mathcal P_{Z,n}^{Q}\widetilde Z_n^{\mathrm{train}}(x).
\end{aligned}
\]
Here, $Z_{\theta^{\mathrm{sh}}}^{\mathrm{sh}}$ learns a representation shared across time levels, $\mathcal N_{\psi_n}^{Z}$ learns the residual specific to level $n$, $\eta_{Z,n}\geq0$ controls the scale of the local component, and $\mathcal P_{Z,n}^{Q}$ imposes a verified control equivariance on the final output. The local head is initialized with zero output so that training at the current level starts from the shared representation and gradually absorbs level-specific information.

When time level $n$ is trained, the shared parameters continue to change. Future-layer functions reconstructed from the live shared trunk may therefore drift away from the snapshots stored at previously completed time levels. To stabilize backward training, let $\mathcal I_n\subset\{n+1,\ldots,N-1\}$ be the replay-layer set and define
\[
Z_{j\leftarrow n}^{\mathrm{live}}(x):=\mathcal P_{Z,j}^{Q}\!\left[
Z_{\theta^{\mathrm{sh}}}^{\mathrm{sh}}(t_j,x)+\eta_{Z,j}\mathcal N_{\psi_j^{\mathrm{snap}}}^{Z}(x)
\right].
\]
The replay-distillation loss is
\[
\mathcal L_{\mathrm{rep},n}:=\sum_{j\in\mathcal I_n}\omega_{n,j} \E_{x\sim\nu_j^{\mathrm{rep}}}
\left[ \left|Z_{j\leftarrow n}^{\mathrm{live}}(x)-\sg\!\left(Z_j^{\mathrm{snap}}(x)\right)\right|^2\right],
\]
where $\omega_{n,j}\geq0$. This loss constrains the live reconstruction of the shared trunk on completed time levels and thereby suppresses cross-layer representational drift. The frozen snapshots themselves remain unchanged and continue to define the functional targets in the backward recursion.

The initial control regression minimizes
\begin{equation}
\begin{aligned}
    \mathcal L_{Z,n}^{\mathrm{init}}=\mathcal L_{Z,n}^{M} +\lambda_{\mathrm{rep},n}\mathcal L_{\mathrm{rep},n},\quad
    \mathcal L_{Z,n}^{M}=\E_{x\sim\nu_n^{\mathrm{tr}}}
      \left[\left|Z_n^{\mathrm{train}}(x)-\sg\!\left(\widehat Z_n^M(x)\right)\right|^2 \right].
\end{aligned}
\label{eq:z_loss}
\end{equation}
Here $\lambda_{\mathrm{rep},n}\geq0$. Control labels are refreshed at a fixed frequency so that the network continues to observe new samples of the conditional moment. At the end of this stage, the current composite control function is deep-copied as $Z_n^{\mathrm{init}}$, which provides a fixed control input for the implicit value target.

% ============================================================================
\subsection{A Unified Implicit Value Operator and Value Regression}
\label{subsec:u_estimation}

% ============================================================================
\subsubsection{Population targets, finite-branch targets, and fixed-point solution}

After the control regression, the current value is determined jointly by the conditional mean of the successor value and $Z_n^{\mathrm{init}}$. Given a successor function $V$ and a current control function $\bar Z$, define the two-input implicit value operator $\mathfrak U_n[V,\bar Z]$ by
\[
\begin{aligned}
    \mathfrak U_n[V,\bar Z](x)&=(P_nV)(x)+h_nf\!\left(t_n,x,\mathfrak U_n[V,\bar Z](x),\bar Z(x)\right).
\end{aligned}
\]
When $h_nL_y<1$, this scalar implicit equation is contractive in the value variable and therefore admits a unique solution. The population target for the first value regression is $\mathfrak U_n\bigl[U_{n+1}^{\mathrm{snap}},Z_n^{\mathrm{init}}\bigr].$
The two-input notation allows the initial value regression and the subsequent value polishing to use the same operator: the successor value function governs the conditional-mean term, whereas the current control function enters the generator.

Given $K_U$ Brownian increments $\boldsymbol w=(w^{(1)},\ldots,w^{(K_U)})$, define the finite-branch operator $\widehat{\mathfrak U}_{n,K_U}[V,\bar Z;\boldsymbol w]$ as the unique solution of
\begin{equation}
\begin{aligned}
    \widehat{\mathfrak U}_{n,K_U}[V,\bar Z;\boldsymbol w](x)
    =\frac1{K_U}\sum_{\ell=1}^{K_U}V\!\left(\Phi_n(x,w^{(\ell)})\right)+h_nf\!\left(
      t_n,x,\widehat{\mathfrak U}_{n,K_U}[V,\bar Z;\boldsymbol w](x), \bar Z(x)\right).
\end{aligned}
\label{eq:value_target}
\end{equation}
The initial value label is
\[
    \widehat Y_n^{\mathrm{init}}(x)
    :=\widehat{\mathfrak U}_{n,K_U}
      \bigl[U_{n+1}^{\mathrm{snap}},Z_n^{\mathrm{init}};
      \boldsymbol w_{n,U}\bigr](x),
\]
where $\boldsymbol w_{n,U}$ is independent of the branches used for the control labels. Because of the implicit nonlinearity, the finite-branch value label is generally not exactly unbiased. Its deviation from the population value target is nevertheless controlled by the sampling error of the conditional mean:
\begin{equation}
\begin{aligned}
&\left|\widehat{\mathfrak U}_{n,K_U}[V,\bar Z;\boldsymbol w](x)-\mathfrak U_n[V,\bar Z](x)\right|\\
&\qquad\leq\frac{1}{1-h_nL_y}\left|\frac1{K_U}\sum_{\ell=1}^{K_U}V\!\left(\Phi_n(x,w^{(\ell)})\right)
      -(P_nV)(x)\right|.
\end{aligned}
\label{eq:value_target_mc_stability}
\end{equation}
Increasing $K_U$ therefore improves the Monte Carlo accuracy of the successor conditional mean directly, and this improvement is transferred to the implicit value target through the factor $(1-h_nL_y)^{-1}$.

Let
\[
    \widehat m_{n,K_U}[V](x):=\frac1{K_U}\sum_{\ell=1}^{K_U}V\!\left(\Phi_n(x,w^{(\ell)})\right).
\]
Starting from $y^{(0)}=\widehat m_{n,K_U}[V](x)$, solve the finite-branch implicit equation through the Picard iteration
\begin{equation}
    y^{(r+1)} =\widehat m_{n,K_U}[V](x) +h_nf\!\left(t_n,x,y^{(r)},\bar Z(x)\right).
\label{eq:picard_iteration}
\end{equation}
With a fixed maximum iteration count $J_{\mathrm{fp}}$ and tolerance $\varepsilon_{\mathrm{fp}}$, convergence is declared when the fixed-point residual $r_{\mathrm{fp}}:=\left|y^{(r)}-\widehat m_{n,K_U}[V](x)-h_nf(t_n,x,y^{(r)},\bar Z(x))\right|$ satisfies $r_{\mathrm{fp}}\leq\varepsilon_{\mathrm{fp}}(1+|y^{(r)}|)$.

% ============================================================================
\subsubsection{Explicit-baseline--neural-residual value model}

The value model separates dominant low-order variation from the remaining nonlinear correction:
\[
\begin{aligned}
    \widetilde U_n^{\mathrm{train}}(x) =\beta_{n,0}+\overline{\bm\varphi}_n(x)^\top\beta_n+s_{r,n}\mathcal N_{\phi_n}^{U}(x),\quad
    U_n^{\mathrm{train}}(x) =\mathcal P_{U,n}\widetilde U_n^{\mathrm{train}}(x).
\end{aligned}
\]
Here, $\overline{\bm\varphi}_n$ is a feature vector after fixed preprocessing, $\beta_{n,0}$ is an intercept, $\beta_n$ contains explicit feature coefficients, $s_{r,n}$ is a stop-gradient scale estimated from the training residual, $\mathcal N_{\phi_n}^{U}$ is a neural residual initialized with zero output, and $\mathcal P_{U,n}$ imposes a verified value invariance.

The value model follows the generic representation ``explicit baseline $+$ neural residual.'' The explicit baseline captures the dominant low-dimensional variation, while the neural residual represents the remaining nonlinear structure. Validated equation-informed features and value projections may be incorporated through the common structural interface; otherwise, generic or empty features and the identity projection are used. By reducing the residual scale, this decomposition may also improve the stability of the automatically differentiated gradient used in control coordination.

Given training targets $\{(x_i,\widehat Y_n^{\mathrm{init}}(x_i))\}_{i=1}^{N_n^{\mathrm{tr}}}$, the ridge baseline is determined by
\[
\begin{aligned}
  (\widehat\beta_{n,0},\widehat\beta_n)\in\operatorname*{argmin}_{b_0,b}
  \Bigg\{&\frac1{N_n^{\mathrm{tr}}}\sum_{i=1}^{N_n^{\mathrm{tr}}}\left|
    b_0+\overline{\bm\varphi}_n(x_i)^\top b-\widehat Y_n^{\mathrm{init}}(x_i)\right|^2
  +\lambda_{\mathrm{ridge},n}|b|^2\Bigg\},
\end{aligned}
\]
where the intercept is not penalized. After fitting, $(\widehat\beta_{n,0},\widehat\beta_n)$ is frozen, and the neural residual is trained separately so that the explicit baseline and nonlinear correction perform distinct approximation tasks.

% ============================================================================
\subsubsection{Value--control compatibility regularization}

The control function is supervised directly by a Brownian conditional moment, whereas the value function is supervised directly by an implicit conditional-mean target. At the population level these two targets satisfy the Markovian relation. Finite branch counts and finite function classes, however, may cause their learned approximations to disagree. Define the current control induced by a value function as
\[
\mathcal G_n[V](x):=\mathcal P_{Z,n}^{Q}\left[\sigma(t_n,x)^\top\nabla_xV(x)\right].
\]
The gradient is taken with respect to the physical state $x$. If the network input contains standardization, logarithmic transformations, or other state transformations, automatic differentiation applies the complete chain rule and returns $\nabla_xV$.

After the ridge baseline has been frozen, the neural residual is trained by minimizing
\begin{equation}
\begin{aligned}
    \mathcal L_{U,n}^{\mathrm{init}}
    &=\E_{x\sim\nu_n^{\mathrm{tr}}}
      \left[\left|U_n^{\mathrm{train}}(x)-\sg\!\left(\widehat Y_n^{\mathrm{init}}(x)\right)\right|^2
      \right]\\
    &\quad+\lambda_{UZ,n}\E_{x\sim\nu_n^{\mathrm{tr}}}
      \left[\left|\mathcal G_n[U_n^{\mathrm{train}}](x)-\sg\!\left(Z_n^{\mathrm{init}}(x)\right)\right|^2
      \right],
\end{aligned}
\label{eq:u_loss}
\end{equation}
which yields $U_n^{\mathrm{init}}$. The first term fits the implicit value target. The second coordinates the value gradient and current control without replacing the direct Brownian supervision of the control. The parameter $\lambda_{UZ,n}$ balances the two objectives. Setting $\lambda_{UZ,n}=0$ recovers value regression without a within-layer compatibility constraint.

% ============================================================================
\subsection{Within-Layer \texorpdfstring{$U$--$Z$}{U--Z} Coordination, Value Calibration, and Snapshot Freezing}
\label{subsec:uz_correction}

After the initial control and value are learned from Brownian conditional-moment and implicit conditional-mean targets, respectively, within-layer coordination combines direct Brownian supervision with the smooth value-gradient proxy through variance-aware control refitting, followed by value polishing under the final control and constant-level calibration.

% ============================================================================
\subsubsection{Variance-aware control refitting}

Construct the stop-gradient proxy from the initial value function:
\[
Z_n^{\nabla U}(x):=\sg\!\left(\mathcal G_n[U_n^{\mathrm{init}}](x)\right).
\]
The finite-sample Monte Carlo fluctuation of $\widehat Z_n^M$ can be estimated directly, whereas $Z_n^{\nabla U}$ inherits the approximation and differentiation errors of the value network. We define a gating weight from the relative fluctuation of the Brownian label:
\[
\gamma_n(x):=\operatorname{clip}\!\left(
      \frac{\widehat v_n^M(x)}{\widehat v_n^M(x)+\tau_\gamma^2(|\widehat Z_n^M(x)|+\varepsilon_Z)^2},
      \gamma_{\min},\gamma_{\max}
    \right),
\]
where  $\operatorname{clip}(r,a,b):=\min\{b,\max\{a,r\}\},$ $\tau_\gamma>0$, $\varepsilon_Z>0$, and $0\leq\gamma_{\min}\leq\gamma_{\max}\leq1$. When the relative fluctuation of the Brownian label is large, $\gamma_n(x)$ increases and the refitting stage gives more weight to the value-gradient proxy. When the label is stable, direct Brownian supervision receives greater weight.

Initialize the control function from $Z_n^{\mathrm{init}}$, freeze $U_n^{\mathrm{init}}$, and minimize
\begin{equation}
\begin{aligned}
    \mathcal L_{Z,n}^{\mathrm{corr}}
    &=\E_{x\sim\nu_n^{\mathrm{tr}}}
      \Bigl[ (1-\gamma_n(x))\left|Z_n^{\mathrm{train}}(x)-\sg\!\left(\widehat Z_n^M(x)\right)\right|^2\\[-1mm]
    &\hspace{16mm} +\gamma_n(x)\left|Z_n^{\mathrm{train}}(x)-Z_n^{\nabla U}(x)\right|^2
      \Bigr]+\lambda_{\mathrm{rep},n}^{\mathrm{corr}}\mathcal L_{\mathrm{rep},n},
\end{aligned}
\label{eq:z_refit}
\end{equation}
to obtain $Z_n^{\mathrm{final}}$. This stage performs a state-dependent soft fusion of the direct control label and the value-gradient proxy, while replay continues to preserve the shared representation across completed time levels.

% ============================================================================
\subsubsection{Value polishing under the final control}

The first value target corresponds to $Z_n^{\mathrm{init}}$. After control refitting has produced $Z_n^{\mathrm{final}}$, the value target is rebuilt so that the final stored value function and the control function actually used by the subsequent backward recursion correspond to the same two-input implicit operator. Using fresh independent branches, define
\begin{equation}
    \widehat Y_n^{\mathrm{corr}}(x):=\widehat{\mathfrak U}_{n,K_U^{\mathrm{corr}}}
      \bigl[U_{n+1}^{\mathrm{snap}},Z_n^{\mathrm{final}};\boldsymbol w_{n,U}^{\mathrm{corr}}
      \bigr](x).
\label{eq:value_correction_target}
\end{equation}
Keep the feature set, preprocessing, and ridge coefficients fixed, and update only the neural residual for a short budget by minimizing
\begin{equation}
\begin{aligned}
    \mathcal L_{U,n}^{\mathrm{corr}}&=\E_{x\sim\nu_n^{\mathrm{tr}}}
      \left[
        \left|U_n^{\mathrm{train}}(x)-\sg\!\left(\widehat Y_n^{\mathrm{corr}}(x)\right)\right|^2
      \right]\\
    &\quad+\lambda_{UZ,n}^{\mathrm{corr}}\E_{x\sim\nu_n^{\mathrm{tr}}}
      \left[
        \left|\mathcal G_n[U_n^{\mathrm{train}}](x)-\sg\!\left(Z_n^{\mathrm{final}}(x)\right)\right|^2
      \right].
\end{aligned}
\label{eq:u_correction_loss}
\end{equation}
Denote the resulting function by $U_n^{\mathrm{corr}}$. If control refitting is disabled, set $Z_n^{\mathrm{final}}=Z_n^{\mathrm{init}}$, and the value-polishing stage may be skipped. A short polishing stage may still be performed under the same control when one wishes only to continue optimizing the current implicit target.

% ============================================================================
\subsubsection{Constant value calibration and immutable snapshots}

Finite-branch conditional means, the intercept of the explicit baseline, and short-budget optimization may leave an approximately constant level bias in the value function. Constant calibration changes only the function level and leaves $\nabla_xU_n$ unchanged. It therefore does not directly alter the control induced by the value gradient. Let $\mathcal D_n^{\mathrm{bias}}=\{x_{n,i}^{\mathrm{bias}}\}_{i=1}^{N_n^{\mathrm{bias}}}$ be an independent set of bias-calibration states. For each state, construct a value target $\widehat Y_{n,i}^{\mathrm{bias}}$ of the same implicit form using fresh Brownian branches and $Z_n^{\mathrm{final}}$, and estimate

\begin{equation}
\begin{aligned}
a_n^\star&:= \operatorname*{arg\,min}_{a \in \mathbb{R}} \frac{1}{N_n^{\mathrm{bias}}}\sum_{i=1}^{N_n^{\mathrm{bias}}}\Bigl|U_n^{\mathrm{corr}}\bigl(x_{n,i}^{\mathrm{bias}}\bigr)
        + a- \widehat{Y}_{n,i}^{\mathrm{bias}}\Bigr|^2\\
&= \frac{1}{N_n^{\mathrm{bias}}}\sum_{i=1}^{N_n^{\mathrm{bias}}}
    \Bigl[\widehat{Y}_{n,i}^{\mathrm{bias}}- U_n^{\mathrm{corr}}\bigl(x_{n,i}^{\mathrm{bias}}\bigr)\Bigr],
\\[2mm]
U_n^{\mathrm{final}}(x)&:= U_n^{\mathrm{corr}}(x) + a_n^\star .
\end{aligned}
\label{eq:constant_bias_correction}
\end{equation}

When the function class or symmetry type requires a zero constant component, set $a_n^\star=0$.

After the current layer has been completed, deep-copy the network parameters, preprocessing, feature transformations, and projection operators of $U_n^{\mathrm{final}}$ and $Z_n^{\mathrm{final}}$ to obtain $(U_n^{\mathrm{snap}},Z_n^{\mathrm{snap}})$, which are used at the next backward time level.

% ============================================================================
\subsection{Complete Algorithm}
\label{subsec:algorithm}

The inputs are the PDE data $(\mu,\sigma,f,g)$, the one-step transition maps $\{\Phi_n\}_{n=0}^{N-1}$, the time grid $\pi$, outer-state design measures, inner-branch budgets, a common adaptive structural interface, and network and optimization parameters. The outputs are the frozen value family $\{U_n^{\mathrm{snap}}\}_{n=0}^{N}$ and the frozen control family $\{Z_n^{\mathrm{snap}}\}_{n=0}^{N-1}$, with $U_N^{\mathrm{snap}}=g$. Algorithm~\ref{alg:dcbsde} presents the complete layerwise backward procedure.

\begin{algorithm}[H]
\caption{Layerwise backward training of \DCBSDE{}}
\label{alg:dcbsde}
\scriptsize
\begin{algorithmic}[1]
\setlength{\itemsep}{0pt}

\Require PDE data $(\mu,\sigma,f,g)$; one-step transition maps $\{\Phi_n\}_{n=0}^{N-1}$; time grid $\pi=\{t_n\}_{n=0}^{N}$;
state-design measures $\{\nu_n^{\mathrm{tr}},\nu_n^{\mathrm{rep}}, \nu_n^{\mathrm{bias}},\nu_n^{\mathrm{val}}\}$;
control branch budget $M=2K$; value branch budgets $K_U$ and $K_U^{\mathrm{corr}}$; a unified equation-informed structural screening interface;
network architectures; 
and optimization parameters.

\Ensure Frozen function families $\{U_n^{\mathrm{snap}}\}_{n=0}^{N}$ and $\{Z_n^{\mathrm{snap}}\}_{n=0}^{N-1}$.

\State Fix the outer-state designs, inner random streams, label-refresh frequencies, fixed-point rules, and stopping criteria for all training stages.

\State Through the unified equation-informed structural screening interface, select and configure the candidate value projections $\mathcal P_{U,n}$, control projections $\mathcal P_{Z,n}^{Q}$, value features $\bm\varphi_n$, and response baselines $B_n^{\mathrm{prob}}$.
\Comment{Adaptive structural configuration}

\State For candidates that are unavailable or not retained, set $\mathcal P_{U,n}=I$ and $\mathcal P_{Z,n}^{Q}=I$, and use generic or empty value features together with an automatic-differentiation or zero response baseline.

\State Initialize the terminal snapshot $U_N^{\mathrm{snap}}\gets g$, the shared control trunk, the time-local heads, and the replay buffer.

\For{$n=N-1,N-2,\ldots,0$}

    \State Use the frozen successor snapshot $U_{n+1}^{\mathrm{snap}}$ and the admissible future-control snapshots.
    \Comment{Discrete backward backbone}

     \State Select $B_n\in\{B_n^{\mathrm{prob}},B_n^{\mathrm{lin}},B_n^{\mathrm{fut}},0\}$ for $n\leq N-2$, and $B_{N-1}\in\{B_{N-1}^{\mathrm{prob}},B_{N-1}^{\mathrm{grad}},B_{N-1}^{\mathrm{lin}},0\}$ for $n=N-1$.

    \State For $x\sim\nu_n^{\mathrm{tr}}$, draw $K$ independent Brownian increments, pair each increment with its negative, and construct $(\widehat Z_n^M(x),\widehat v_n^M(x))$ using~\eqref{eq:antithetic_target} and~\eqref{eq:target_variance}.
    \Comment{Finite-branch control labels}

    \State Minimize~\eqref{eq:z_loss} using the shared temporal trunk, time-local residual head, control projection, and replay distillation.

    \State Store a deep copy of the trained current-layer control as $Z_n^{\mathrm{init}}$.
    \Comment{Initial control regression}

    \State Using Brownian branches independent of those used for the control labels, construct the finite-branch implicit value equation
    in~\eqref{eq:value_target}.

    \State Solve the implicit equation by~\eqref{eq:picard_iteration} to obtain the initial within-layer value labels
    $\widehat Y_n^{\mathrm{init}}$.

    \State Fit and freeze the explicit value baseline, and then minimize~\eqref{eq:u_loss} to obtain $U_n^{\mathrm{init}}$.
    \Comment{Initial value regression}

    \If{control refitting is enabled}

        \State Compute the state-dependent gate $\gamma_n$ from $\widehat v_n^M$, and construct the value-gradient proxy $Z_n^{\nabla U}$ from $U_n^{\mathrm{init}}$.

        \State Minimize~\eqref{eq:z_refit} to obtain $Z_n^{\mathrm{final}}$.

    \Else

        \State Set $Z_n^{\mathrm{final}}\gets Z_n^{\mathrm{init}}$.

    \EndIf

    \If{value polishing under the final control is enabled}

        \State Using fresh independent branches, reconstruct the value target under $Z_n^{\mathrm{final}}$
        according to~\eqref{eq:value_correction_target}.

        \State Minimize~\eqref{eq:u_correction_loss} to obtain $U_n^{\mathrm{corr}}$.

    \Else

        \State Set $U_n^{\mathrm{corr}}\gets U_n^{\mathrm{init}}$.

    \EndIf

    \If{constant-bias calibration is enabled and the function class permits a nonzero constant component}

        \State Estimate $a_n^\star$ on independent calibration states using~\eqref{eq:constant_bias_correction}.

        \State Set
        $U_n^{\mathrm{final}}\gets U_n^{\mathrm{corr}}+a_n^\star$.

    \Else

        \State Set $a_n^\star\gets0$ and $U_n^{\mathrm{final}}\gets U_n^{\mathrm{corr}}$.

    \EndIf

    \State Deep-copy the final value and control functions together with their preprocessing, feature transformations, and projection operators:
    \[
        U_n^{\mathrm{snap}}\gets\operatorname{DeepCopy}\bigl(U_n^{\mathrm{final}}\bigr),\qquad
        Z_n^{\mathrm{snap}}\gets\operatorname{DeepCopy}\bigl(Z_n^{\mathrm{final}}\bigr).
    \]

    \State Update the replay buffer and continue to the next backward time level.

\EndFor

\State \Return
$\{U_n^{\mathrm{snap}}\}_{n=0}^{N}$ and $\{Z_n^{\mathrm{snap}}\}_{n=0}^{N-1}$.

\end{algorithmic}
\end{algorithm}

For a deterministic initial state $x_0$, the algorithm returns $U_0^{\mathrm{snap}}(x_0)$ and $Z_0^{\mathrm{snap}}(x_0)$. Pathwise outputs are obtained by evaluating $\{U_n^{\mathrm{snap}},Z_n^{\mathrm{snap}}\}$ along the corresponding discrete forward states. 
% ============================================================================
\subsection{Error Decomposition, Backward Stability, and Complexity Analysis}
\label{subsec:theoretical_analysis}

% ============================================================================
\subsubsection{Separating discretization error from local algorithmic error}
\label{subsubsec:error_decomposition}

Let $(U_n^\theta,Z_n^\theta)$ denote the frozen functions after all training, coordination, and calibration stages at the current time level, and set
\[
    z_n^\pi=\mathcal M_nu_{n+1}^\pi,\qquad
    z(t_n,x)=\sigma(t_n,x)^\top\nabla_xu(t_n,x).
\]
The control error can be decomposed as
\begin{equation}
\begin{aligned}
Z_n^\theta-z(t_n,\cdot)={}
\underbrace{\left(Z_n^\theta-\mathcal M_nU_{n+1}^\theta\right)}_{\text{current-layer control residual}}
+\underbrace{\mathcal M_n\left(U_{n+1}^\theta-u_{n+1}^\pi\right)}_{\text{propagated successor-value error}} 
+\underbrace{\left(z_n^\pi-z(t_n,\cdot)\right)}_{\text{time-discretization error}}.
\end{aligned}
\label{eq:control_error_decomposition}
\end{equation}
Likewise, using the two-input value operator,
\begin{equation}
\begin{aligned}
U_n^\theta-u(t_n,\cdot)={}&\underbrace{\left(
U_n^\theta-\mathfrak U_n[U_{n+1}^\theta,Z_n^\theta]\right)}_{\text{current-layer value residual}} 
+\underbrace{\left(\mathfrak U_n[U_{n+1}^\theta,Z_n^\theta]-\mathfrak U_n[u_{n+1}^\pi,z_n^\pi]
\right)}_{\text{propagated successor-value and control errors}}\\
&+\underbrace{\left(u_n^\pi-u(t_n,\cdot)\right)}_{\text{time-discretization error}}.
\end{aligned}
\label{eq:value_error_decomposition}
\end{equation}
Here $u_n^\pi=\mathfrak U_n[u_{n+1}^\pi,z_n^\pi]$. Equations~\eqref{eq:control_error_decomposition}--\eqref{eq:value_error_decomposition} separate three sources of error: the time-discretization error between the continuous problem and the discrete recursion, propagation of successor-function errors through the exact discrete operators, and local residuals generated at the current level by finite branching, function approximation, optimization, and coordination.

% ============================================================================
\subsubsection{Brownian projection contraction and stability of the exact backward backbone}
\label{subsubsec:projection_stability}

Define
\[
    \mathcal Z_n[V]:=\mathcal M_nV,\qquad
    \mathcal U_n[V]:=\mathfrak U_n[V,\mathcal Z_n[V]].
\]
Thus, $\mathcal U_n$ inserts the successor value function into both the conditional-mean term and the ideal discrete control, and is the exact discrete backward operator for the implicit Euler scheme.

\begin{lemma}[Conditional contraction of the Brownian projection]
\label{lem:brownian_projection}
For every square-integrable function $e:\R^d\to\R$,
\begin{equation}
    h_n\left|(\mathcal M_ne)(x)\right|^2\leq(P_n|e|^2)(x)-|(P_ne)(x)|^2.
\label{eq:brownian_projection}
\end{equation}
\end{lemma}

\begin{proof}
Let $\bar e_x(w):=e(\Phi_n(x,w))-(P_ne)(x)$. Since $\E_w[w]=0$,
\[
    (\mathcal M_ne)(x)=\frac1{h_n}\E_w[\bar e_x(w)w].
\]
For any unit vector $a\in\R^m$, the Cauchy--Schwarz inequality and $\E_w[(a^\top w)^2]=h_n$ imply
\[
    |a^\top(\mathcal M_ne)(x)|^2\leq\frac1{h_n}\E_w|\bar e_x(w)|^2.
\]
Taking the supremum over $a$ and using
$\E_w|\bar e_x(w)|^2=(P_n|e|^2)(x)-|(P_ne)(x)|^2$ yields the result.
\end{proof}

This contraction converts the $h_n^{-1}$ Brownian weight into a conditional-variance term. 

\begin{proposition}[One-step stability of the exact discrete backward operator]
\label{prop:backward_stability}
Assume that the generator satisfies
\[
    |f(t,x,y,z)-f(t,x,y',z')|\leq L_y|y-y'|+L_z|z-z'|,
\]
and that $h_nL_y<1$. Then, for any square-integrable successor functions $V,V'$,
\begin{equation}
    h_n
    |\mathcal Z_n[V](x)-\mathcal Z_n[V'](x)|^2\leq(P_n|V-V'|^2)(x),
\label{eq:control_operator_stability}
\end{equation}
and
\begin{equation}
    |\mathcal U_n[V](x)-\mathcal U_n[V'](x)|^2\leq\alpha_n(P_n|V-V'|^2)(x),
\label{eq:value_operator_stability}
\end{equation}
where
\[
\alpha_n:=\frac{1+h_nL_z^2}{(1-h_nL_y)^2}.
\]
\end{proposition}

\begin{proof}
Let $e=V-V'$. The control estimate follows directly from Lemma~\ref{lem:brownian_projection}. For the value operator, its definition and the Lipschitz continuity of the generator give
\[
\begin{aligned} &(1-h_nL_y)|\mathcal U_n[V](x)-\mathcal U_n[V'](x)| \leq |(P_ne)(x)|+h_nL_z|(\mathcal M_ne)(x)|.
\end{aligned}
\]
Set $a:=|(P_ne)(x)|$ and $b:=\sqrt{h_n}|(\mathcal M_ne)(x)|$. Lemma~\ref{lem:brownian_projection} yields
$a^2+b^2\leq(P_n|e|^2)(x)$, whereas $(a+\sqrt{h_n}L_zb)^2 \leq(1+h_nL_z^2)(a^2+b^2).$
Combining the two inequalities proves~\eqref{eq:value_operator_stability}.
\end{proof}

For any scalar- or vector-valued function $F$ with finite second moment, define $\|F\|_{n,2}^2:=\E\!\left[|F(\bar X_n)|^2\right],$ where the expectation is also taken over algorithmic randomness when $F$ is random.
Taking expectations at $\bar X_n$ and iterating backward gives
\[
\|U_n-U_n'\|_{n,2}^2\leq\exp(C(T-t_n))\|U_N-U_N'\|_{N,2}^2.
\]

% ============================================================================
\subsubsection{Backward stability under local residuals}
\label{subsubsec:residual_stability}

The exact-operator estimate describes the propagation of successor errors. To cover finite-branch regression and neural training, define the normalized local residuals of the final frozen functions relative to the population discrete recursion:
\begin{equation}
\begin{aligned}
    \mathscr R_{U,n}^\theta(x)&:=\frac{1}{h_n}\Bigl(U_n^\theta(x)-(P_nU_{n+1}^\theta)(x)
      -h_nf(t_n,x,U_n^\theta(x),Z_n^\theta(x))\Bigr),\\
    \mathscr R_{Z,n}^\theta(x)&:=Z_n^\theta(x)-(\mathcal M_nU_{n+1}^\theta)(x).
\end{aligned}
\label{eq:normalized_local_residuals}
\end{equation}
After division by $h_n$, the value residual has the same scale as the generator. Consequently, the local value and control errors both enter the global energy estimate with an $h_n$ weight.

\begin{proposition}[Local-residual-driven backward stability]
\label{prop:residual_backward_stability}
Under Assumption~\textup{(H2)}, there exist constants $h_0>0$, $c_0>0$, and $C_0>0$ depending only on $L_y$ and $L_z$ such that, whenever $|\pi|\leq h_0$, the errors
\[
    e_n:=U_n^\theta-u_n^\pi,\qquad
    q_n:=Z_n^\theta-z_n^\pi
\]
satisfy the one-step estimate
\begin{equation}
\begin{aligned}
    |e_n(x)|^2+c_0h_n|q_n(x)|^2
    \leq{}&(1+C_0h_n)(P_n|e_{n+1}|^2)(x) +C_0h_n\Bigl(|\mathscr R_{U,n}^\theta(x)|^2
      +|\mathscr R_{Z,n}^\theta(x)|^2\Bigr).
\end{aligned}
\label{eq:one_step_residual_stability}
\end{equation}
Moreover, there exists a constant $C>0$ depending only on $T,L_y,L_z$ such that, for every $0\leq n\leq N$,
\begin{equation}
\begin{aligned}
&\|e_n\|_{n,2}^2+c_0\sum_{j=n}^{N-1}h_j\|q_j\|_{j,2}^2\leq
C\Biggl[\|e_N\|_{N,2}^2+\sum_{j=n}^{N-1}h_j
\left(\|\mathscr R_{U,j}^\theta\|_{j,2}^2+\|\mathscr R_{Z,j}^\theta\|_{j,2}^2\right)
\Biggr].
\end{aligned}
\label{eq:global_residual_stability}
\end{equation}
\end{proposition}

\begin{proof}
Fix a time level and state $x$, abbreviate $h=h_n$, and set
\[
    m=(P_ne_{n+1})(x),\qquad
    \zeta=(\mathcal M_ne_{n+1})(x),\qquad
    p=(P_n|e_{n+1}|^2)(x).
\]
By the exact discrete recursion and~\eqref{eq:normalized_local_residuals},
\begin{equation}
\begin{aligned}
    e_n(x)=m+h\,\delta f_n(x)+h\mathscr R_{U,n}^\theta(x),\qquad
    q_n(x)=\zeta+\mathscr R_{Z,n}^\theta(x),
\end{aligned}
\label{eq:residual_error_recursion}
\end{equation}
where
\[
\delta f_n(x):= f(t_n,x,U_n^\theta(x),Z_n^\theta(x))-f(t_n,x,u_n^\pi(x),z_n^\pi(x))
\]
satisfies $|\delta f_n|\leq L_y|e_n|+L_z|q_n|$. Lemma~\ref{lem:brownian_projection} gives
\begin{equation}
    p\geq |m|^2+h|\zeta|^2.
\label{eq:projection_energy_pair}
\end{equation}

Let $\kappa:=1+12L_z^2$. When $\kappa h<1$, Young's inequality yields
\begin{align*}
|m|^2&=\left|e_n-h(\delta f_n+\mathscr R_{U,n}^\theta)\right|^2\\
&\geq(1-\kappa h)|e_n|^2-\frac{h}{\kappa}|\delta f_n+\mathscr R_{U,n}^\theta|^2\\
&\geq\left(1-\kappa h-\frac{3L_y^2}{\kappa}h\right)|e_n|^2-\frac{3L_z^2}{\kappa}h|q_n|^2
-\frac{3}{\kappa}h|\mathscr R_{U,n}^\theta|^2.
\end{align*}
On the other hand, $\zeta=q_n-\mathscr R_{Z,n}^\theta$, and hence
\[
    h|\zeta|^2\geq \frac{1}{2}h|q_n|^2-h|\mathscr R_{Z,n}^\theta|^2.
\]
Since $3L_z^2/\kappa\leq1/4$, substituting these two bounds into~\eqref{eq:projection_energy_pair} gives
\[
\begin{aligned}
 p\geq{}&(1-Ah)|e_n|^2+\frac{1}{4}h|q_n|^2 -C_1h\left(
 |\mathscr R_{U,n}^\theta|^2+|\mathscr R_{Z,n}^\theta|^2
 \right),
\end{aligned}
\]
where $A=\kappa+3L_y^2/\kappa$, and $C_1$ depends only on $L_y,L_z$. Choose $h_0$ such that $Ah_0\leq1/2$ and absorb $(1-Ah)^{-1}$ into $1+C_0h$. This proves~\eqref{eq:one_step_residual_stability}.

Taking expectations at $x=\bar X_n$ and using
$\E[(P_n|e_{n+1}|^2)(\bar X_n)]=\|e_{n+1}\|_{n+1,2}^2$ gives
\[
\begin{aligned}
\|e_n\|_{n,2}^2+c_0h_n\|q_n\|_{n,2}^2
\leq{}&(1+C_0h_n)\|e_{n+1}\|_{n+1,2}^2+C_0h_n\left(
\|\mathscr R_{U,n}^\theta\|_{n,2}^2+\|\mathscr R_{Z,n}^\theta\|_{n,2}^2
\right).
\end{aligned}
\]
Multiplying by the discrete integrating factor $\prod_{k<n}(1+C_0h_k)$, summing from $n$ to $N-1$, and using
$\prod_{k=n}^{j}(1+C_0h_k)\leq\exp(C_0(T-t_n))$ proves~\eqref{eq:global_residual_stability}.
\end{proof}

Proposition~\ref{prop:residual_backward_stability} reduces the global error of the complete algorithm to the terminal error and the population-recursion residual at each time level. Replay, structured parameterization, compatibility regularization, control refitting, value polishing, and constant calibration all enter the same stability bound through the final residuals $\mathscr R_{U,n}^\theta$ and $\mathscr R_{Z,n}^\theta$.

% ============================================================================
\subsubsection{Local learning errors and finite-branch errors}
\label{subsubsec:local_branch_errors}

To connect the stability estimate more directly with the actual training targets, let $\widehat Z_n^M$ be the finite-branch control label constructed from $U_{n+1}^\theta$. Let $\widehat Y_n^{\mathrm{br}}$ denote the finite-branch implicit value target solved exactly under the final control $Z_n^\theta$, and let $\widehat Y_n^{\mathrm{cmp}}$ denote the target actually returned by the fixed-point iteration. Define the control-side errors
\begin{equation}
\begin{aligned}
    \ell_{Z,n}:=Z_n^\theta-\widehat Z_n^M,\qquad
    \xi_{Z,n}:=\widehat Z_n^M-\mathcal M_nU_{n+1}^\theta,
\end{aligned}
\label{eq:control_target_error_split}
\end{equation}
and the normalized value-side errors
\begin{equation}
\begin{aligned}
    \ell_{U,n}&:=\frac{U_n^\theta-\widehat Y_n^{\mathrm{cmp}}}{h_n},
    &\varepsilon_{\mathrm{fp},n}&:=\frac{\widehat Y_n^{\mathrm{cmp}}-\widehat Y_n^{\mathrm{br}}}{h_n},\\
    \xi_{U,n}&:=\frac{\widehat Y_n^{\mathrm{br}}-
    \mathfrak U_n[U_{n+1}^\theta,Z_n^\theta]}{h_n}.
\end{aligned}
\label{eq:value_target_error_split}
\end{equation}
Here, $\ell_{Z,n}$ and $\ell_{U,n}$ collect the local learning discrepancy remaining after function approximation, empirical-risk generalization, optimization, replay, and within-layer coordination. The term $\varepsilon_{\mathrm{fp},n}$ is the numerical fixed-point error, whereas $\xi_{Z,n}$ and $\xi_{U,n}$ are the finite-Brownian-branch errors relative to the corresponding population targets.

By definition,
\begin{equation}
    \mathscr R_{Z,n}^\theta=\ell_{Z,n}+\xi_{Z,n}.
\label{eq:z_residual_target_split}
\end{equation}
Next let $Y_n^{\mathrm{pop}}:=\mathfrak U_n[U_{n+1}^\theta,Z_n^\theta]$. Since
$Y_n^{\mathrm{pop}}=P_nU_{n+1}^\theta+h_nf(t_n,\cdot,Y_n^{\mathrm{pop}},Z_n^\theta)$,
\begin{equation}
\begin{aligned}
|\mathscr R_{U,n}^\theta|&\leq(1+h_nL_y)\left|\frac{U_n^\theta-Y_n^{\mathrm{pop}}}{h_n}\right|\\
&\leq(1+h_nL_y)\left(|\ell_{U,n}|+|\varepsilon_{\mathrm{fp},n}|+|\xi_{U,n}|\right).
\end{aligned}
\label{eq:u_residual_target_split}
\end{equation}

When the labels and trained outputs contain algorithmic randomness, the norms below average over both the discrete state chain and the algorithmic randomness. Define the aggregated local learning and numerical error by
\begin{equation}
\begin{aligned}
\mathcal E_{\mathrm{loc}}(\pi):={}&\|U_N^\theta-g\|_{N,2}^2 +\sum_{n=0}^{N-1}h_n
\left(\|\ell_{Z,n}\|_{n,2}^2+\|\ell_{U,n}\|_{n,2}^2+\|\varepsilon_{\mathrm{fp},n}\|_{n,2}^2
\right),
\end{aligned}
\label{eq:aggregate_local_error}
\end{equation}
and the finite-branch error by
\begin{equation}
    \mathcal E_{\mathrm{br}}(\pi):=\sum_{n=0}^{N-1}h_n
    \left(\|\xi_{Z,n}\|_{n,2}^2+\|\xi_{U,n}\|_{n,2}^2
    \right).
\label{eq:aggregate_branch_error}
\end{equation}

\begin{corollary}[Learning--branch decomposition of the error to the discrete solution]
\label{cor:discrete_algorithm_error}
Under the conditions of Proposition~\ref{prop:residual_backward_stability}, there exists a constant $C>0$, independent of the grid and the network parameters, such that
\begin{equation}
\begin{aligned}
&\max_{0\leq n\leq N}\|U_n^\theta-u_n^\pi\|_{n,2}^2+\sum_{n=0}^{N-1}h_n\|Z_n^\theta-z_n^\pi\|_{n,2}^2 \leq C\left(\mathcal E_{\mathrm{loc}}(\pi)+\mathcal E_{\mathrm{br}}(\pi)
\right).
\end{aligned}
\label{eq:discrete_algorithm_bound}
\end{equation}
\end{corollary}

\begin{proof}
Equations~\eqref{eq:z_residual_target_split}--\eqref{eq:u_residual_target_split} and the finite-sum inequality imply
\[
\sum_{n=0}^{N-1}h_n\left(\|\mathscr R_{U,n}^\theta\|_{n,2}^2+\|\mathscr R_{Z,n}^\theta\|_{n,2}^2
\right)
\leq C\left(\mathcal E_{\mathrm{loc}}(\pi)+\mathcal E_{\mathrm{br}}(\pi)\right).
\]
Substituting this estimate into~\eqref{eq:global_residual_stability} and maximizing over the starting level $n$ proves the result.
\end{proof}

The finite-branch terms are controlled directly by the label analysis in Section~\ref{sec:preliminaries}. On the control side,~\eqref{eq:antithetic_mc_variance} gives the conditional mean-square error of $\xi_{Z,n}$. On the value side,~\eqref{eq:value_target_mc_stability} yields
\begin{equation}
    |\xi_{U,n}(x)|
    \leq\frac{1}{h_n(1-h_nL_y)}\left|\widehat m_{n,K_U}[U_{n+1}^\theta](x)-(P_nU_{n+1}^\theta)(x)
    \right|.
\label{eq:normalized_value_branch_bound}
\end{equation}
Thus, the number of control branches determines the accuracy of the Brownian conditional moment, while the number of value branches determines the accuracy of the conditional mean on a one-time-step scale. Together, the two branch budgets govern $\mathcal E_{\mathrm{br}}(\pi)$.

% ============================================================================
\subsubsection{Consistency under vanishing discretization, local learning, and finite-branch errors}
\label{subsubsec:consistency}

Define the discretization error of the implicit Euler scheme relative to the continuous Markovian solution by
\begin{equation}
\begin{aligned}
\mathcal E_{\mathrm{disc}}(\pi):={}&\max_{0\leq n\leq N}
\|u_n^\pi-u(t_n,\cdot)\|_{n,2}^2 +\sum_{n=0}^{N-1}h_n\|z_n^\pi-z(t_n,\cdot)\|_{n,2}^2.
\end{aligned}
\label{eq:discretization_error_energy}
\end{equation}
Under Assumptions~\textup{(H1)}--\textup{(H2)} and the corresponding temporal regularity, consistency of the selected implicit Euler scheme means that $\mathcal E_{\mathrm{disc}}(\pi)\to0$.

\begin{corollary}[Total consistency when all three error classes vanish]
\label{cor:total_consistency}
Consider a sequence of \DCBSDE{} approximations with $|\pi|\to0$. Suppose that the stability constants in Proposition~\ref{prop:residual_backward_stability} remain uniformly bounded. Then
\begin{equation}
\begin{aligned}
&\max_{0\leq n\leq N}
\|U_n^\theta-u(t_n,\cdot)\|_{n,2}^2+\sum_{n=0}^{N-1}h_n\|Z_n^\theta-z(t_n,\cdot)\|_{n,2}^2\\
&\qquad\leq C\left(\mathcal E_{\mathrm{disc}}(\pi)+\mathcal E_{\mathrm{loc}}(\pi)+\mathcal E_{\mathrm{br}}(\pi)
\right).
\end{aligned}
\label{eq:total_error_bound}
\end{equation}
In particular, if
\begin{equation}
    \mathcal E_{\mathrm{disc}}(\pi)+\mathcal E_{\mathrm{loc}}(\pi)+\mathcal E_{\mathrm{br}}(\pi)
    \longrightarrow0,
\label{eq:three_error_vanishing}
\end{equation}
then
\begin{equation}
\begin{aligned}
    \max_{0\leq n\leq N}\|U_n^\theta-u(t_n,\cdot)\|_{n,2}^2
    +\sum_{n=0}^{N-1}h_n\|Z_n^\theta-z(t_n,\cdot)\|_{n,2}^2 \longrightarrow0.
\end{aligned}
\label{eq:conditional_consistency}
\end{equation}
\end{corollary}

\begin{proof}
For both value and control, add and subtract the discrete solution $(u_n^\pi,z_n^\pi)$, apply the squared triangle inequality, and combine Corollary~\ref{cor:discrete_algorithm_error} with~\eqref{eq:discretization_error_energy}. This gives~\eqref{eq:total_error_bound}. Equation~\eqref{eq:conditional_consistency} then follows from~\eqref{eq:three_error_vanishing}.
\end{proof}

Equation~\eqref{eq:total_error_bound} makes explicit the route to conditional consistency of the learned approximation. Grid refinement controls $\mathcal E_{\mathrm{disc}}$. Function-class approximation, empirical-risk generalization, optimization, fixed-point solution, and within-layer coordination control $\mathcal E_{\mathrm{loc}}$. The inner branch budgets for control and value determine $\mathcal E_{\mathrm{br}}$. These three error classes accumulate through the backward energy estimate of Proposition~\ref{prop:residual_backward_stability}, rather than being amplified by a fixed factor at every time level.

% ============================================================================
\subsubsection{Computational complexity and storage requirements}
\label{subsubsec:complexity}

We separate the complexity into target construction, network optimization, and snapshot storage. Let $N_{Z,n}^{\mathrm{out}}$, $N_{U,n}^{\mathrm{out}}$, and $N_{C,n}^{\mathrm{out}}$ be the numbers of outer states used at level $n$ for the control labels, initial value targets, and value-correction targets, respectively. Let $R_{Z,n}$, $R_{U,n}$, and $R_{C,n}$ be the corresponding numbers of target refreshes. Let $K_U$, $K_U^{\mathrm{corr}}$, and $K_U^{\mathrm{bias}}$ be the inner branch counts used in the three value stages. The dominant number of evaluations of the frozen successor value function is
\[
\begin{aligned}
\mathcal C_{\mathrm{succ}}=\sum_{n=0}^{N-1}
\Bigl(&R_{Z,n}M N_{Z,n}^{\mathrm{out}}+R_{U,n}K_U N_{U,n}^{\mathrm{out}}+R_{C,n}K_U^{\mathrm{corr}}N_{C,n}^{\mathrm{out}}
+K_U^{\mathrm{bias}}N_n^{\mathrm{bias}}
\Bigr).
\end{aligned}
\]
Here $M=2K$ already counts both positive and negative Brownian branches. If validation diagnostics call the successor value function again, those evaluations should be added in the same way. Fixed-point solution is performed after the branch average has been formed. If each target at level $n$ uses, on average, $J_{\mathrm{fp},n}$ generator evaluations, the additional number of generator calls is
\[
\mathcal C_{\mathrm{fp}}=O\!\left(
      \sum_{n=0}^{N-1}J_{\mathrm{fp},n}
      \bigl(R_{U,n}N_{U,n}^{\mathrm{out}}+R_{C,n}N_{C,n}^{\mathrm{out}}+N_n^{\mathrm{bias}}\bigr)
    \right).
\]

To describe the optimization cost, let $I_{Z,n}^{\mathrm{init}}$, $I_{Z,n}^{\mathrm{corr}}$, $I_{U,n}^{\mathrm{init}}$, and $I_{U,n}^{\mathrm{corr}}$ be the numbers of gradient steps in the four training stages. Let $B_{Z,n}$ and $B_{U,n}$ be the corresponding batch sizes. Let $C_Z(P_{Z,n})$ and $C_U(P_{U,n})$ denote the per-sample forward--backward costs for control and value networks with parameter counts $P_{Z,n}$ and $P_{U,n}$. Let $C_{\nabla U}(P_{U,n})$ be the automatic-differentiation cost of the physical-state gradient, and let $C_{\mathrm{rep},n}$ be the cost of one replay batch. Ignoring hardware parallelism and constant factors, the network-training cost is
\begin{equation}
\begin{aligned}
\mathcal C_{\mathrm{train}}=O\!\Bigg(
\sum_{n=0}^{N-1}\Bigl\{
&I_{Z,n}^{\mathrm{init}}B_{Z,n}C_Z(P_{Z,n})+I_{Z,n}^{\mathrm{corr}}B_{Z,n}C_Z(P_{Z,n})\\
&+I_{U,n}^{\mathrm{init}}B_{U,n}\bigl[C_U(P_{U,n})+C_{\nabla U}(P_{U,n})\bigr]\\
&+I_{U,n}^{\mathrm{corr}}B_{U,n}\bigl[C_U(P_{U,n})+C_{\nabla U}(P_{U,n})\bigr]+\left(I_{Z,n}^{\mathrm{init}}+I_{Z,n}^{\mathrm{corr}}\right)
 C_{\mathrm{rep},n}
\Bigr\}
\Bigg).
\end{aligned}
\label{eq:training_complexity}
\end{equation}
When control refitting, value polishing, the gradient-compatibility term, or replay is disabled, the corresponding terms are removed from~\eqref{eq:training_complexity}.

Let $S(U_n^{\mathrm{snap}})$ and $S(Z_n^{\mathrm{snap}})$ denote the storage required by the complete immutable value and control snapshots at level $n$. Let $S_{\mathrm{live}}$ denote the storage for the live shared trunk, local heads, and optimizer state, and let $S_{\mathrm{cache}}$ denote replay and nested-branch caches. The final storage and peak working storage satisfy
\begin{equation}
\begin{aligned}
    S_{\mathrm{final}}=O\!\left(
      \sum_{n=0}^{N}S(U_n^{\mathrm{snap}})+\sum_{n=0}^{N-1}S(Z_n^{\mathrm{snap}})+S_{\mathrm{meta}}
    \right),\\
    S_{\mathrm{peak}}=O\!\left(S_{\mathrm{live}}+\max_{0\leq n<N}\left\{S(U_{n+1}^{\mathrm{snap}})+S_{\mathrm{cache},n}\right\}\right),
\end{aligned}
\label{eq:storage_complexity}
\end{equation}
where $S_{\mathrm{meta}}$ includes state preprocessing, feature selection, and structural-projection information. If all snapshots remain resident on the accelerator, their total storage must also be added to the peak device-memory requirement. If snapshots are stored in host memory or on disk and loaded one layer at a time, the main additional cost is data transfer rather than a larger number of simultaneously trainable parameters.

Relative to methods that do not construct explicit control labels, the principal additional costs of \DCBSDE{} are Brownian-branch evaluations, a separate control regression, and immutable snapshots.

% ============================================================================
% ============================================================================
\section{Numerical Experiments}
\label{sec:experiments}

This section evaluates \DCBSDE{} numerically. We first compare end-to-end accuracy and then assess the learned value process $U$, control process $Z$, and discrete dynamic consistency. We finally examine the roles of the main algorithmic components and whether the observed differences can be attributed solely to training time.

% ============================================================================
\subsection{Experimental Setting and Evaluation Criteria}
\label{subsec:experimental_design}

The experiments cover six benchmark problems. Each problem--method combination uses $n=10$ independent random seeds, paired across methods within each benchmark. For paired descriptive bootstrap intervals, the common seed indices are resampled jointly within each benchmark so that the method pairing is preserved. Appendix~\ref{app:benchmark_definitions} specifies the forward coefficients $(\mu,\sigma)$, generator $f$, terminal condition $g$, reference information, and evaluation endpoints. Table~\ref{tab:benchmark_protocol_revised} summarizes the main-comparison configurations.

\begin{table}[H]
\centering
\caption{Numerical configurations and prespecified primary metrics.}
\label{tab:benchmark_protocol_revised}
\small
\resizebox{\textwidth}{!}{%
\begin{tabular}{@{}llcccll@{}}
\toprule
Benchmark class & Configuration & $d$ & $T$ & $N$ & Reference information & Primary metric \\
\midrule
HJB & HJB-Quadratic & 50 & 1.0 & 40 & Analytic path reference & $U$-path RMSE \\
HJB & Paper-HJB & 100 & 1.0 & 20 & Monte Carlo initial-value reference & $|U_0-U_{0,\mathrm{ref}}|$ \\
Allen--Cahn & Allen--Cahn & 100 & 0.3 & 20 & Branching-diffusion initial-value reference & $|U_0-U_{0,\mathrm{ref}}|$ \\
Burgers & Burgers--20 & 20 & 1.0 & 80 & Analytic path reference & $U$-path RMSE \\
Quadratic-gradient & Quadratic-gradient & 100 & 1.0 & 30 & Manufactured path reference & $U$-path RMSE \\
Reaction--diffusion & Reaction--diffusion & 100 & 1.0 & 30 & Manufactured path reference & $U$-path RMSE \\
\bottomrule
\end{tabular}}
\end{table}

The classical baselines are DeepBSDE-style, DBDP1, and DBDP2. DeepBSDE-style optimizes the initial value and all time-level controls through a global terminal loss. DBDP1 proceeds backward and jointly approximates $U_n$ and $Z_n$ using a one-step loss, whereas DBDP2 approximates the value function and recovers the control from its spatial gradient. By contrast, at each level \DCBSDE{} constructs an explicit Brownian conditional-moment label from the frozen successor value, trains the control, solves the implicit value target, and performs the within-layer coordination described in Section~\ref{sec:method}.

We additionally compare \DCBSDE{} with DPI, a recent regression-based solver for high-dimensional PDEs \cite{HanHuLongZhao2026}. At each Picard iteration, DPI generates value and gradient labels from Feynman--Kac and variance-reduced Bismut--Elworthy--Li representations and fits a space--time network through gradient-augmented regression.

Within each benchmark, all five methods use the same forward discretization, paired training and test seeds, and evaluation grid. Network architectures, optimizers, batch sizes, learning rates, training steps, stopping rules, and method-specific budgets were fixed before the formal runs. Reference solutions and test metrics were reserved for final evaluation. Experiments were conducted on five NVIDIA GeForce RTX~5090 GPUs using Python~3.12.3, PyTorch~2.8.0, and CUDA~12.8. Each training run used one GPU, with independent benchmark--method--seed jobs distributed across devices; wall-clock runs were executed without competing jobs on the same GPU.

For HJB-Quadratic, Burgers--20, Quadratic-gradient, and Reaction--diffusion, pathwise references are available, and the primary metric is
\begin{equation}
\operatorname{RMSE}_{U,\mathrm{path}}=
\left[\frac{1}{N}\sum_{n=0}^{N-1}\mathbb E\left|\widehat U_n(\bar X_n)-u(t_n,\bar X_n)
\right|^2\right]^{1/2}.
\label{eq:exp_u_path_rmse}
\end{equation}
Here, for each evaluated trained method, we denote its returned value and control approximations by
\[
\widehat U_n:=U_n^{\mathrm{snap}}, \qquad \widehat Z_n:=Z_n^{\mathrm{snap}}.
\]

For Allen--Cahn and Paper-HJB, only an independently computed initial-value reference is available, so the primary metric is
\begin{equation}
E_0 = \left|\widehat U_0(x_0)-U_{0,\mathrm{ref}}\right|.
\label{eq:exp_u0_error}
\end{equation}
Because these metrics evaluate different quantities, rankings and ratios are reported only within the same benchmark and under the same metric.

When a reference control path is available, we also report
\begin{equation}
\operatorname{RMSE}_{Z,\mathrm{path}} =
\left[
\frac{1}{N}\sum_{n=0}^{N-1}\mathbb E\left\|
\widehat Z_n(\bar X_n)-\sigma(t_n,\bar X_n)^{\mathsf T} \nabla_xu(t_n,\bar X_n)
\right\|_2^2
\right]^{1/2}.
\label{eq:exp_z_path_rmse}
\end{equation}

Two reference-free diagnostics assess the discrete dynamics. The one-step BSDE residual is
\begin{equation}
\begin{aligned}
\widehat r_n={}&\widehat U_{n+1}(\bar X_{n+1})-\widehat U_n(\bar X_n)\\
&+h_nf\!\left(t_n,\bar X_n,\widehat U_n(\bar X_n),\widehat Z_n(\bar X_n)\right)-\widehat Z_n(\bar X_n)^{\mathsf T}\Delta W_n,
\end{aligned}
\label{eq:exp_one_step_residual}
\end{equation}
The corresponding pathwise residual RMSE is
\begin{equation}
\operatorname{RMSE}_{\mathrm{res}}:=\left[\frac1N\sum_{n=0}^{N-1}\E|\widehat r_n|^2\right]^{1/2}.
\label{eq:exp_one_step_residual_rmse}
\end{equation}

For the terminal-rollout diagnostic, initialize $\widetilde Y_0=\widehat U_0(\bar X_0)$ and propagate
\begin{equation}
\widetilde Y_{n+1} = \widetilde Y_n -h_nf\!\left( t_n,\bar X_n,\widetilde Y_n, \widehat Z_n(\bar X_n)
\right) +\widehat Z_n(\bar X_n)^{\mathsf T}\Delta W_n.
\label{eq:exp_terminal_rollout}
\end{equation}
The corresponding RMSE is
\begin{equation}
\operatorname{RMSE}_{\mathrm{roll}}:=\left(\E\left|\widetilde Y_N-g(\bar X_N)\right|^2\right)^{1/2}.
\label{eq:exp_terminal_rollout_rmse}
\end{equation}

The one-step residual measures local discrete consistency, whereas the terminal rollout measures accumulated discrepancy.

% ============================================================================
\subsection{Benchmark Comparison under the Primary Metrics}
\label{subsec:main_matrix_results_revised}

\medskip 
\noindent\textit{(Comparison with established deep BSDE solvers.}
Tables~\ref{tab:main_path_results_revised} and~\ref{tab:main_u0_results_revised} report the primary value metrics. Entries are the mean $\pm$ sample standard deviation of the corresponding error metric over $n=10$ paired seeds. Boldface indicates the smallest observed mean in each row. Across the six configurations, \DCBSDE{} attains the smallest observed mean in five. On Burgers--20, DeepBSDE-style yields a slightly smaller mean $U$-path RMSE than \DCBSDE{} ($0.0091$ versus $0.0108$).

\begin{table}[H]
\centering
\caption{Prespecified primary value errors, $U-path RMSE$, for the four benchmarks with pathwise value references.}
\label{tab:main_path_results_revised}
\small
\resizebox{\textwidth}{!}{%
\begin{tabular}{@{}lcccc@{}}
\toprule
Benchmark
& \DCBSDE{}
& DeepBSDE-style
& DBDP1
& DBDP2 \\
\midrule

HJB-Quadratic
& $\bm{0.1173\pm0.0239}$
& $2.3177\pm0.4904$
& $20.2363\pm0.1260$
& $20.3243\pm0.1323$ \\

Burgers--20
& $0.0108\pm0.0019$
& $\bm{0.0091\pm0.0019}$
& $0.0664\pm0.0099$
& $0.0611\pm0.0091$ \\

Quadratic-gradient
& $\bm{0.0373\pm0.0004}$
& $0.1868\pm0.0040$
& $0.0733\pm0.0046$
& $0.0695\pm0.0041$ \\

Reaction--diffusion
& $\bm{0.0041\pm5.352\times10^{-4}}$
& $0.0129\pm0.0021$
& $0.0667\pm0.0028$
& $0.0608\pm0.0037$ \\

\bottomrule
\end{tabular}}
\end{table}

\begin{table}[H]
\centering
\caption{Prespecified primary value errors, $|U_0-U_{0,\mathrm{ref}}|$, for the two benchmarks with independently computed initial-value references. }
\label{tab:main_u0_results_revised}
\small
\resizebox{\textwidth}{!}{%
\begin{tabular}{@{}lcccc@{}}
\toprule
Benchmark
& \DCBSDE{}
& DeepBSDE-style
& DBDP1
& DBDP2 \\
\midrule

Allen--Cahn
& $\bm{9.940\times10^{-5}\pm1.310\times10^{-5}}$
& $2.926\times10^{-4}\pm2.943\times10^{-4}$
& $2.311\times10^{-4}\pm1.910\times10^{-4}$
& $1.542\times10^{-4}\pm9.951\times10^{-5}$ \\

Paper-HJB
& $\bm{7.544\times10^{-4}\pm2.613\times10^{-4}}$
& $0.0093\pm0.0037$
& $3.9906\pm3.640\times10^{-5}$
& $0.0102\pm0.0047$ \\

\bottomrule
\end{tabular}}
\end{table}

Figure~\ref{fig:strongest_baseline_vs_full} compares \DCBSDE{} with the baseline having the smallest observed mean within each configuration. The ratio is defined as ``best observed baseline/\DCBSDE{}''; values above $1$ favor \DCBSDE{}, whereas values below $1$ favor the selected baseline. Comparisons with all three baselines are reported in Appendix~\ref{app:extended_experiments}.

\begin{figure}[H]
\centering
\includegraphics[width=0.86\textwidth]
{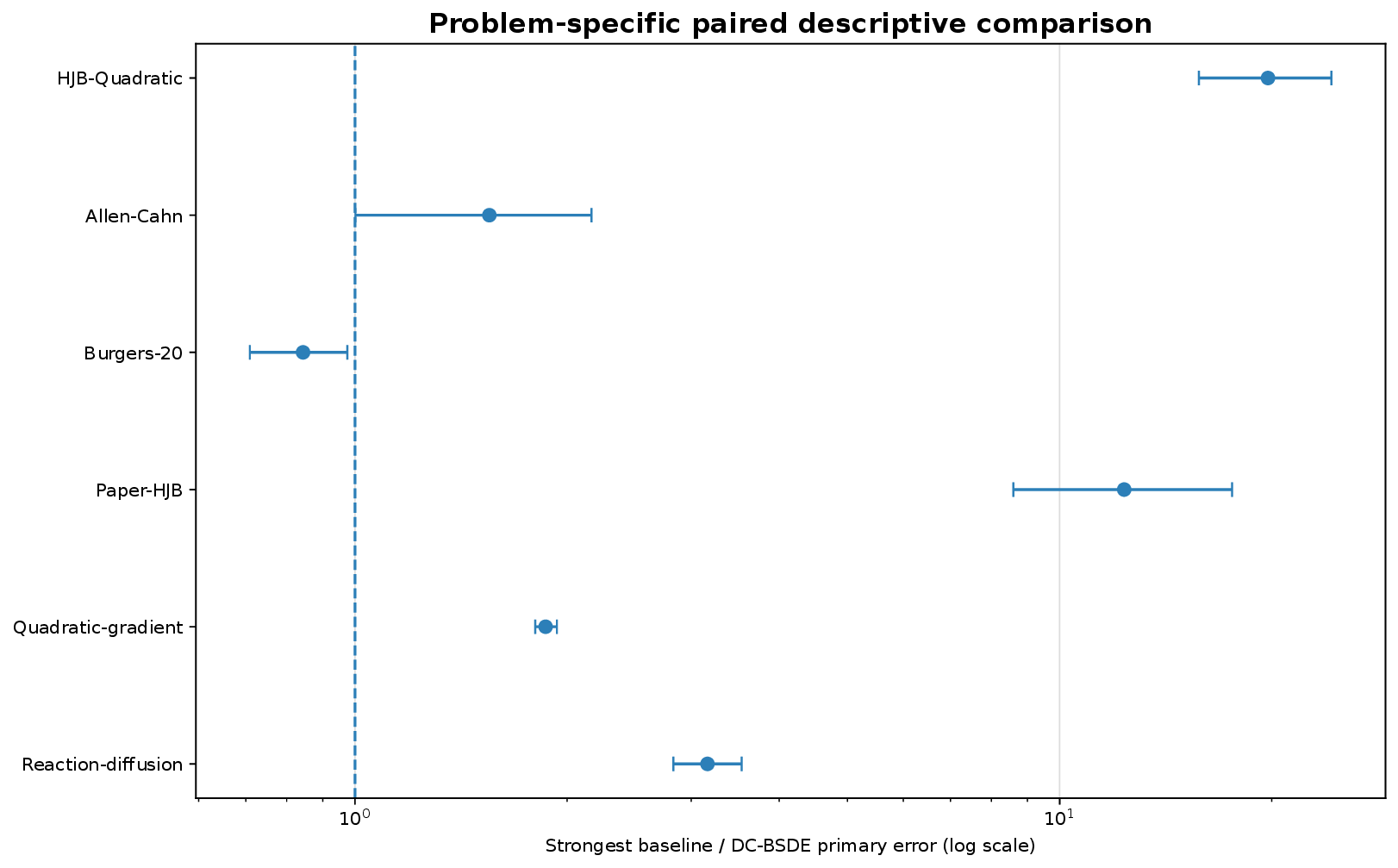}
\caption{Ratio of the primary error of the best observed baseline to that of \DCBSDE{}, with $95\%$ paired descriptive bootstrap intervals}
\label{fig:strongest_baseline_vs_full}
\end{figure}

The largest observed separations occur on HJB-Quadratic and Paper-HJB, where the best-baseline/\DCBSDE{} ratios are $19.753$ and $12.347$, respectively. The corresponding ratios are $1.864$ for Quadratic-gradient, $3.164$ for Reaction--diffusion, and $1.551$ for Allen--Cahn, all favoring \DCBSDE{}.

The ranking is reversed on Burgers--20, where DeepBSDE-style attains a smaller mean value-path error, corresponding to a best-baseline/\DCBSDE{} ratio of $0.844$. Thus, the primary value metric does not establish uniform dominance by any method and motivates the joint examination of control accuracy and discrete dynamic consistency.

The comparatively strong DeepBSDE-style result on Burgers--20 may be attributed to the smooth solution's effectively one-dimensional dependence on the averaged state and to the effectiveness of its global terminal objective in transmitting a stable end-to-end training signal. This advantage is confined to the primary value metric; the control and dynamic diagnostics below reveal a different accuracy trade-off.

\medskip 
\noindent\textit{Comparison with Deep Picard Iteration.}
All $60$ DPI benchmark--seed jobs completed the prescribed computational workflow. However, a post-hoc stability audit identified divergent runs on HJB-Quadratic and Paper-HJB. Table~\ref{tab:dpi_primary_results} reports the DPI pass counts and value errors computed only over runs that passed the audit, whereas the \DCBSDE{} statistics use all ten seeds. 

\begin{table}[H]
\centering
\caption{Descriptive comparison of DPI and \DCBSDE{} under the benchmark-specific primary value metrics.}
\label{tab:dpi_primary_results}
\scriptsize
\resizebox{\textwidth}{!}{%
\begin{tabular}{@{}llcccc@{}}
\toprule
Benchmark
& Primary metric
& DPI passes
& DPI: conditional mean $\pm$ SD
& \DCBSDE{}: mean $\pm$ SD
& DPI/\DCBSDE{} \\
\midrule

HJB-Quadratic
& $U$-path RMSE
& $5/10$
& $11.9935\pm0.4734$
& $\mathbf{0.1173\pm0.0239}$
& $102.217$ \\

Allen--Cahn
& $E_0$
& $10/10$
& $6.828\times10^{-4}\pm1.782\times10^{-4}$
& $\mathbf{9.940\times10^{-5}\pm1.310\times10^{-5}}$
& $6.869$ \\

Paper-HJB
& $E_0$
& $1/10$
& $0.9168$ $(n=1)$
& $\mathbf{7.544\times10^{-4}\pm2.613\times10^{-4}}$
& $1215.209$ \\

Burgers--20
& $U$-path RMSE
& $10/10$
& $\mathbf{0.00298\pm5.015\times10^{-4}}$
& $0.01084\pm0.0019$
& $0.275$ \\

Quadratic-gradient
& $U$-path RMSE
& $10/10$
& $0.0935\pm0.0054$
& $\mathbf{0.0373\pm3.985\times10^{-4}}$
& $2.508$ \\

Reaction--diffusion
& $U$-path RMSE
& $10/10$
& $0.0128\pm0.0011$
& $\mathbf{0.0041\pm5.352\times10^{-4}}$
& $3.135$ \\

\bottomrule
\end{tabular}}
\end{table}

For Allen--Cahn, Burgers--20, Quadratic-gradient, and Reaction--diffusion, all ten DPI runs passed the stability audit, allowing direct comparisons based on the same number of paired seeds. \DCBSDE{} attains smaller primary value errors on Allen--Cahn, Quadratic-gradient, and Reaction--diffusion. On Burgers--20, however, DPI performs better, reducing the mean $U$-path RMSE from $0.01084$ to $0.00298$, or to $27.5\%$ of the \DCBSDE{} error.

DPI exhibits non-convergent runs on both HJB-Quadratic and Paper-HJB, with only five and one of the ten runs, respectively, passing the stability audit. Even among the passing runs, the conditional mean errors of DPI are substantially larger than the corresponding unconditional ten-seed means of \DCBSDE{}, with DPI/\DCBSDE{} ratios of $102.217$ and $1215.209$, respectively.
% ============================================================================
\subsection{Control Accuracy and Dynamic Diagnostics}
\label{subsec:uz_diagnostics_results_revised}

\medskip 
\noindent\textit{Comparison with established deep BSDE solvers.}
Table~\ref{tab:DCBSDE_vs_strongest_diagnostics_revised} reports the control-path error and two reference-free dynamic diagnostics. For each benchmark, the selected baseline is the method with the smallest observed mean primary value error. A dash indicates that no validated reference control path is available. Complete diagnostics for all four methods are reported in Appendix~\ref{app:extended_experiments}.

\begin{table}[H]
\centering
\caption{Mean control and dynamic-consistency diagnostics for \DCBSDE{} and the baseline with the smallest observed mean primary value error.}
\label{tab:DCBSDE_vs_strongest_diagnostics_revised}
\scriptsize
\setlength{\tabcolsep}{3.5pt}
\resizebox{\textwidth}{!}{%
\begin{tabular}{@{}llcccccc@{}}
\toprule
Benchmark
& Selected baseline
& \multicolumn{2}{c}{$Z$-path RMSE}
& \multicolumn{2}{c}{BSDE residual RMSE}
& \multicolumn{2}{c}{Terminal-rollout RMSE} \\
\cmidrule(lr){3-4}
\cmidrule(lr){5-6}
\cmidrule(lr){7-8}
&
& \DCBSDE{} & Baseline
& \DCBSDE{} & Baseline
& \DCBSDE{} & Baseline \\
\midrule

HJB-Quadratic
& DeepBSDE-style
& $\mathbf{0.0768}$ & $0.7567$
& $\mathbf{0.1631}$ & $1.3471$
& $\mathbf{0.9686}$ & $1.5729$ \\

Allen--Cahn
& DBDP2
& \textemdash & \textemdash
& $\mathbf{3.881\times10^{-4}}$ & $2.35\times10^{-3}$
& $\mathbf{0.00152}$ & $0.00520$ \\

Burgers--20
& DeepBSDE-style
& $\mathbf{0.0138}$ & $0.0317$
& $\mathbf{0.00958}$ & $0.0108$
& $0.0481$ & $\mathbf{0.0472}$ \\

Paper-HJB
& DeepBSDE-style
& \textemdash & \textemdash
& $\mathbf{0.00757}$ & $0.0381$
& $\mathbf{0.0337}$ & $0.0518$ \\

Quadratic-gradient
& DBDP2
& $\mathbf{0.0660}$ & $0.1568$
& $\mathbf{0.0161}$ & $0.0947$
& $\mathbf{0.1168}$ & $0.3756$ \\

Reaction--diffusion
& DeepBSDE-style
& $\mathbf{0.00651}$ & $0.0211$
& $\mathbf{0.0112}$ & $0.0141$
& $\mathbf{0.0585}$ & $0.0622$ \\

\bottomrule
\end{tabular}}
\end{table}

Figure~\ref{fig:u_z} displays the joint value--control errors for the four benchmarks with complete pathwise references. Both axes are logarithmic; movement toward the lower-left corner indicates simultaneous reductions in the $U$- and $Z$-path RMSEs.

\begin{figure}[H]
\centering
\includegraphics[width=0.92\textwidth]{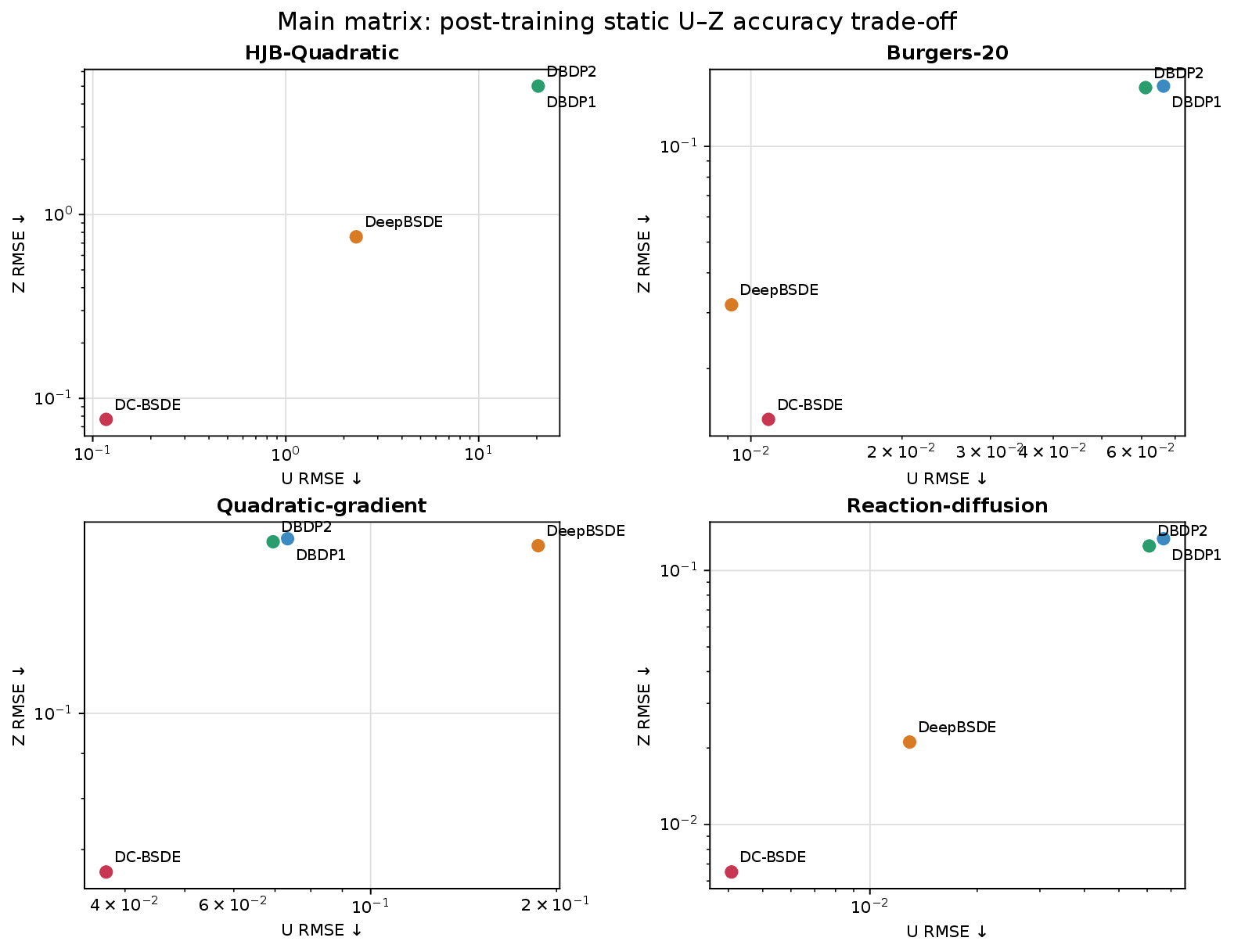}
\caption{Joint $U$-path and $Z$-path errors after training}
\label{fig:u_z}
\end{figure}

On HJB-Quadratic, Quadratic-gradient, and Reaction--diffusion, \DCBSDE{} has smaller mean $U$- and $Z$-path errors than the selected baseline. Its $Z$-path RMSE is approximately $10.1\%$, $42.1\%$, and $30.9\%$ of the corresponding baseline error, respectively. These results are consistent with the intended role of explicit Brownian control supervision: control accuracy can improve independently of the scalar value
error rather than arising only as a by-product of value fitting.

For Allen--Cahn and Paper-HJB, no validated reference control path is available, so no direct conclusion about $Z$ accuracy can be drawn. Nevertheless, \DCBSDE{} yields smaller mean BSDE-residual and terminal-rollout errors than the selected baseline on both benchmarks, supporting improved discrete dynamic consistency.

Burgers--20 exhibits a different trade-off. DeepBSDE-style has the smaller mean $U$-path RMSE, $0.0091$ compared with $0.0108$, whereas the \DCBSDE{} $Z$-path RMSE, $0.0138$, is only $43.6\%$ of the corresponding
DeepBSDE-style error. The two methods therefore occupy different positions on the empirical $U$--$Z$ Pareto front: DeepBSDE-style is preferable under the value-only endpoint, while \DCBSDE{} provides the more accurate control approximation. 

\medskip 
\noindent\textit{Comparison with Deep Picard Iteration.}
The comparison with DPI reveals benchmark-dependent value--control trade-offs. On Burgers--20, all ten DPI runs passed the stability audit, and DPI achieves smaller mean $U$-path RMSE ($0.00298$ versus $0.01084$) and $Z$-path RMSE ($0.0021$ versus $0.0138$) than \DCBSDE{}. On Quadratic-gradient and Reaction--diffusion, \DCBSDE{} has the smaller mean $U$ error, whereas DPI has the smaller mean $Z$ error; the corresponding DPI/\DCBSDE{} ratios for the $Z$-path RMSE are approximately $0.488$ and $0.477$.
On HJB-Quadratic, \DCBSDE{} has smaller mean errors in both $U$ and $Z$ than the DPI runs that passed the stability audit. This comparison is conditional because only five of the ten DPI runs passed the audit, whereas all ten \DCBSDE{} runs completed stably.

Thus, although neither method uniformly dominates the other, \DCBSDE{} exhibits a more consistent overall balance between value and control accuracy across these benchmarks, together with stronger cross-seed stability on the HJB problems.

Figure~\ref{fig:dpi_hjb_stability} illustrates the seed-dependent DPI instability on the HJB benchmarks. The validation losses of runs that failed the audit increase by several orders of magnitude after a seed-dependent number of Picard iterations. Solid curves denote runs passing the post-hoc stability audit, red dashed curves denote runs failing it, and the vertical axes are logarithmic.

\begin{figure}[!htbp]
\centering

\begin{minipage}[t]{0.48\textwidth}
\centering
\includegraphics[width=\linewidth]
{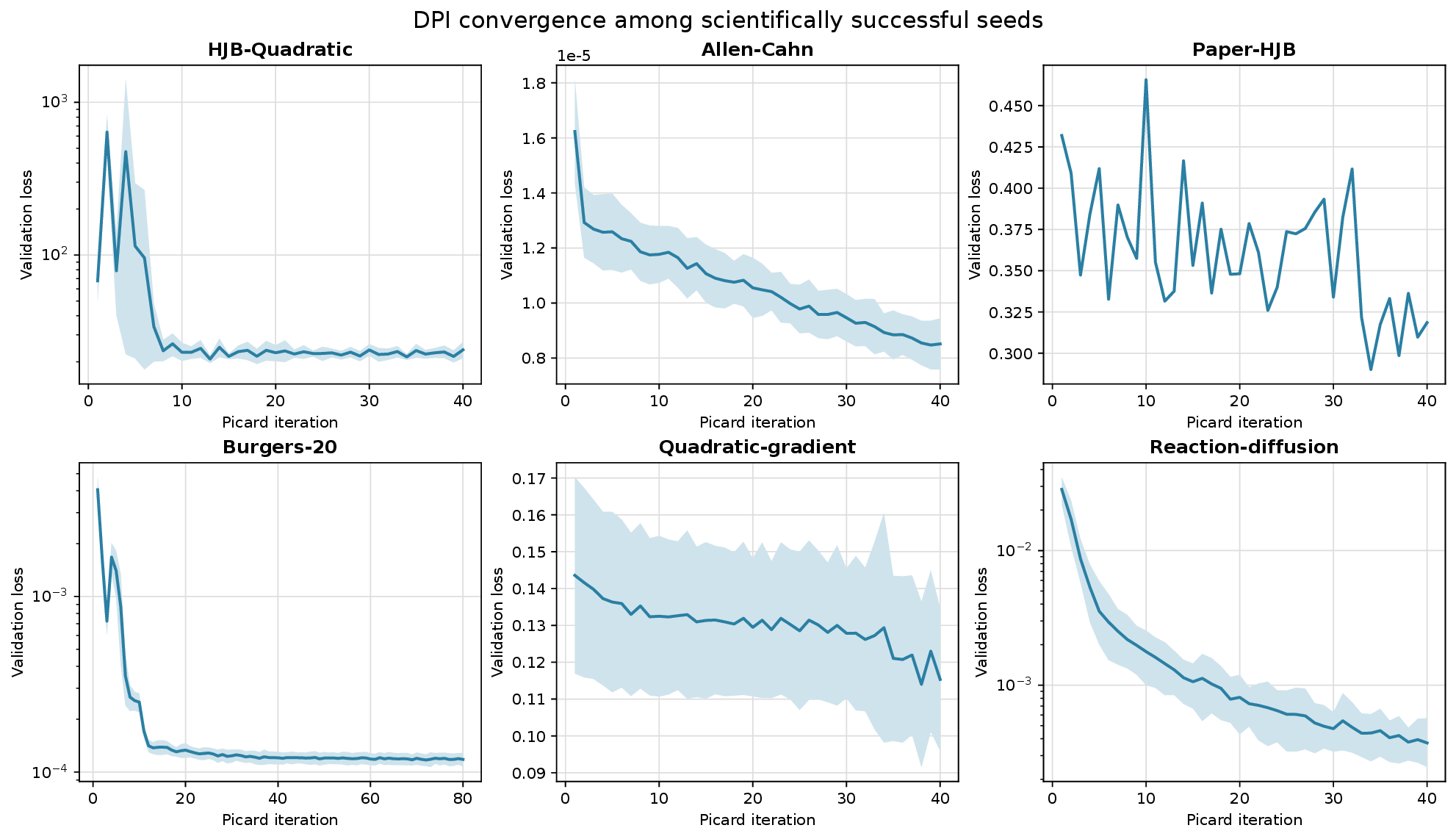}
\end{minipage}
\hfill
\begin{minipage}[t]{0.48\textwidth}
\centering
\includegraphics[width=\linewidth]
{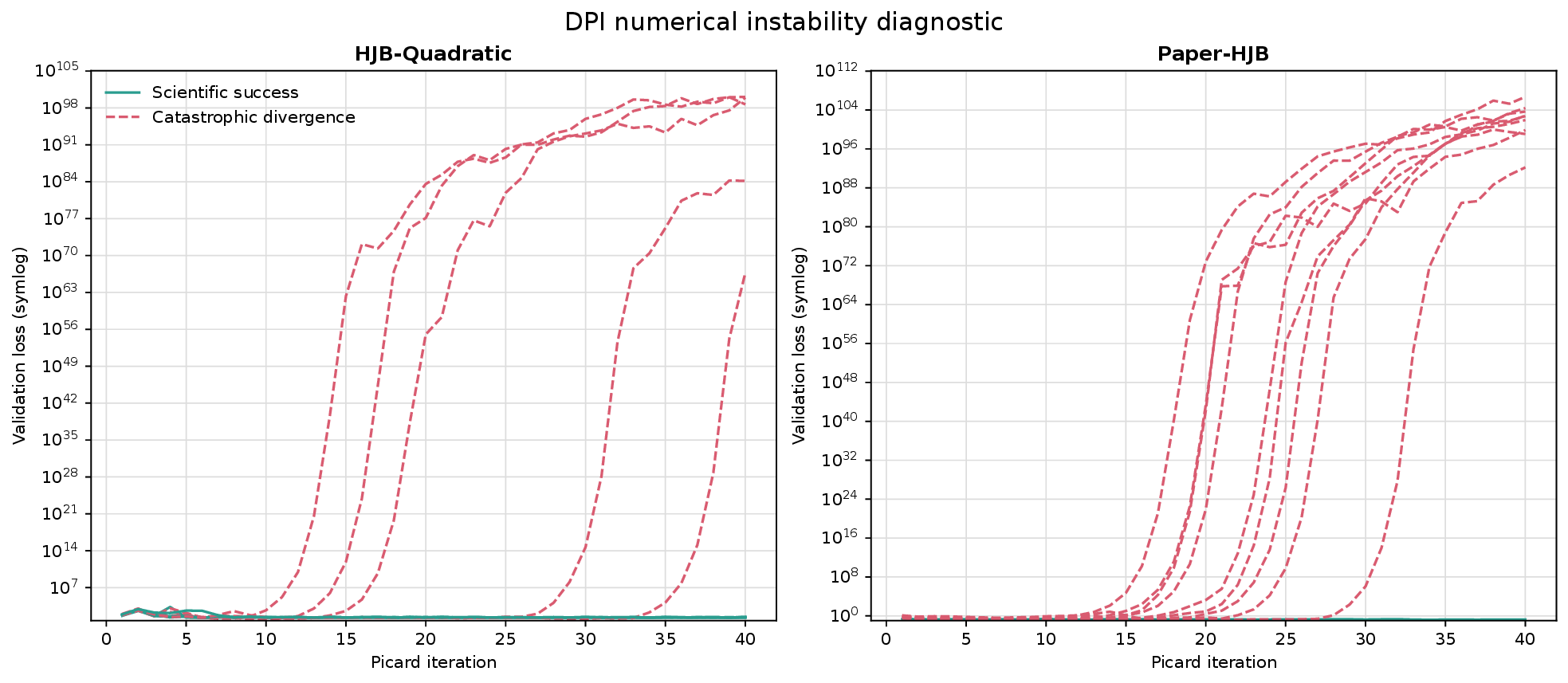}
\end{minipage}

\caption{DPI validation-loss trajectories and divergence diagnostics on the two HJB benchmarks}
\label{fig:dpi_hjb_stability}
\end{figure}

The quadratic Hamiltonian may act as an error-amplification channel during Picard iteration. For $q(z)=-\tfrac12\lVert z\rVert_2^2$, a perturbation $e$ gives $q(z+e)-q(z)=-z^{\mathsf T}e-\tfrac12\lVert e\rVert_2^2$.
Because the fitted network is reused in subsequent iterations, control errors may be repeatedly propagated and amplified. The available DPI theory assumes global Lipschitz continuity and does not provide a uniform contraction bound for this quadratic setting \cite{HanHuLongZhao2026}.

% ============================================================================
\subsection{Component Ablations and Structural-Interface Sensitivity}
\label{subsec:mechanism_results_revised}

We conduct two complementary controlled studies. Component ablations disable selected training, variance-reduction, or within-layer coordination mechanisms in the locked full \DCBSDE{} configuration, whereas the structural-interface study varies the candidate sources supplied to the common screening interface. In both studies, the equation data, discrete backward backbone, prescribed training budget, frozen population control target, and layerwise backward order remain unchanged.

% ============================================================================
\subsubsection{Ablation of Selected Training and Coordination Mechanisms}
\label{subsec:numerical_ablation_results_revised}

Starting from the locked full \DCBSDE{} configuration, we disable selected training and coordination mechanisms. The \emph{local-only} variant removes the shared temporal trunk and its associated replay constraint, retaining independent time-local control representations. The \emph{neural-only $U$} variant suppresses the feature-based explicit branch while keeping the candidate pool and screening outcome fixed. The \emph{no antithetic sampling}, \emph{no control baseline}, and \emph{no $Z$ refit} variants disable antithetic pairing, the frozen linear-response baseline, and control refitting, respectively. 

Quadratic-gradient and Reaction--diffusion both use $d=100$, $T=1$, and $N=30$, with manufactured references for $u$ and $z$ evaluated along common test paths. This permits joint assessment of value error, control error, and the one-step dynamic residual under comparable discretizations. Table~\ref{tab:numerical_ablation_results} reports error ratios relative to the locked full configuration. 

\begin{table}[H]
\centering
\caption{Component-ablation results relative to the locked full \DCBSDE{} configuration.}
\label{tab:numerical_ablation_results}
\scriptsize
\setlength{\tabcolsep}{3.5pt}
\renewcommand{\arraystretch}{1.05}
\begin{tabularx}{\textwidth}{@{}lXcccc@{}}
\toprule
Benchmark
& Ablated variant
& \shortstack{Primary-$U$\\ratio}
& \shortstack{$95\%$ paired\\interval}
& \shortstack{$Z$\\ratio}
& \shortstack{Residual\\ratio} \\
\midrule
Quadratic-gradient
& local-only
& 0.925
& $[0.918,0.932]$
& 2.160
& 2.509 \\
Quadratic-gradient
& neural-only $U$
& 1.878
& $[1.862,1.894]$
& 2.496
& 3.506 \\
Quadratic-gradient
& no antithetic sampling
& 0.940
& $[0.929,0.949]$
& 1.342
& 2.904 \\
Quadratic-gradient
& no control baseline
& 0.932
& $[0.919,0.946]$
& 1.386
& 2.895 \\
Quadratic-gradient
& no $Z$ refit
& 1.030
& $[1.017,1.043]$
& 0.909
& 0.986 \\
\addlinespace[2pt]
Reaction--diffusion
& local-only
& 0.991
& $[0.886,1.109]$
& 73.957
& 6.462 \\
Reaction--diffusion
& neural-only $U$
& 4.827
& $[4.500,5.227]$
& 6.169
& 2.382 \\
Reaction--diffusion
& no antithetic sampling
& 1.203
& $[1.084,1.337]$
& 64.948
& 5.603 \\
Reaction--diffusion
& no control baseline
& 1.135
& $[1.081,1.199]$
& 69.344
& 6.746 \\
Reaction--diffusion
& no $Z$ refit
& 1.658
& $[1.553,1.781]$
& 1.444
& 1.010 \\
\bottomrule
\end{tabularx}
\end{table}

Figure~\ref{fig:numerical_ablations_revised} displays the primary-$U$ error ratios and their $95\%$ paired descriptive bootstrap intervals. The horizontal axis is logarithmic, and the reference line at $1$ separates decreases from increases relative to the full configuration. The corresponding control-error and residual ratios are reported in Table~\ref{tab:numerical_ablation_results}.

\begin{figure}[H]
\centering
\includegraphics[width=0.84\textwidth]{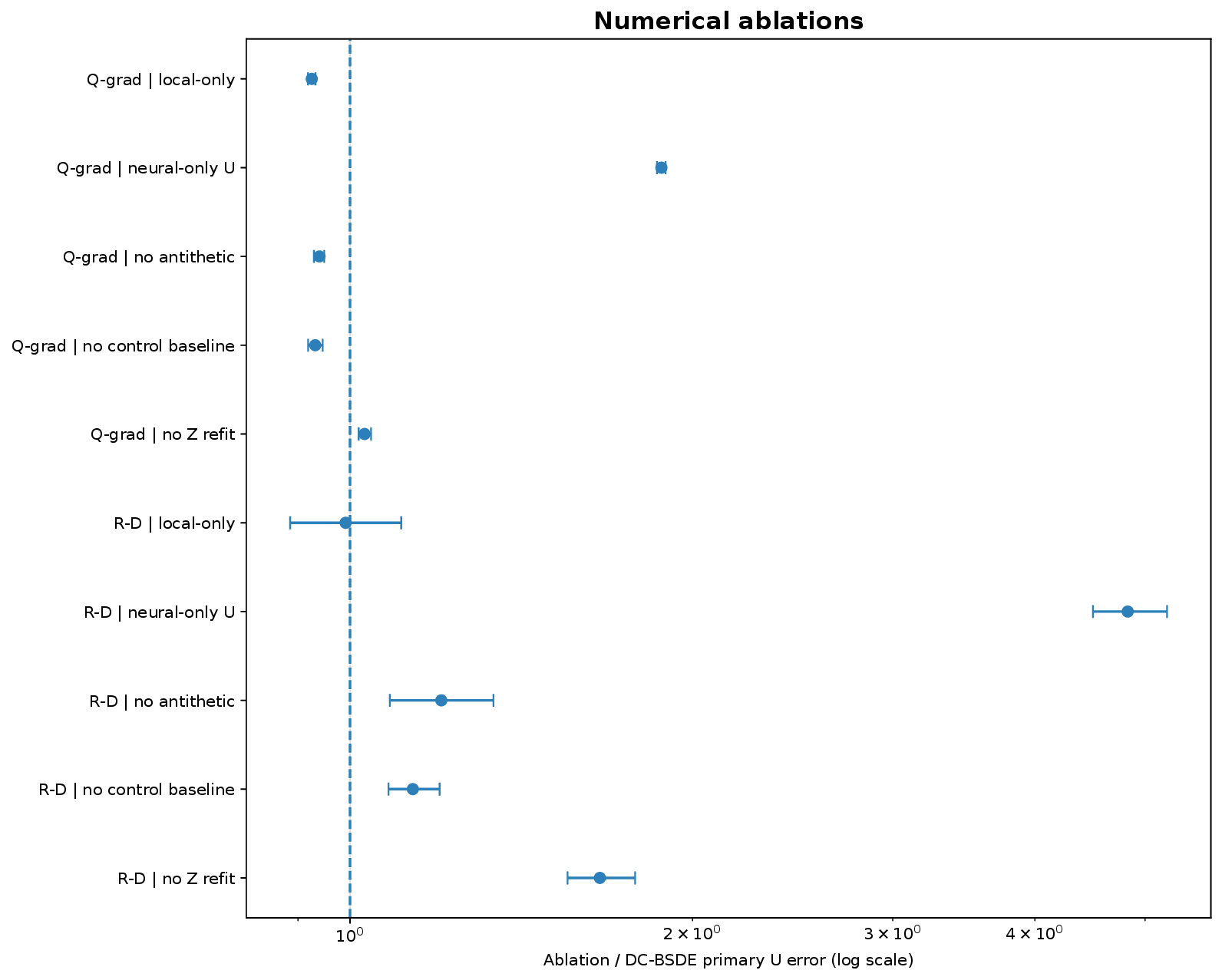}
\caption{Primary-$U$ error ratios for the component ablations}
\label{fig:numerical_ablations_revised}
\end{figure}

\medskip 
\noindent\textit{Feature-based value representation.}
Suppressing the feature-based explicit branch increases the primary $U$ error by factors of $1.878$ and $4.827$ on Quadratic-gradient and Reaction--diffusion, respectively, while the corresponding $Z$-error ratios
increase to $2.496$ and $6.169$. This joint deterioration is consistent with the ``explicit baseline $+$ neural residual'' representation reducing the approximation burden on the neural residual and supporting the final control approximation through value--control coordination. 

\medskip 
\noindent\textit{Shared--local control representation and replay.}
On Quadratic-gradient, the local-only variant slightly reduces the primary $U$ error but increases the $Z$ error and dynamic residual by factors of $2.160$ and $2.509$. On Reaction--diffusion, its primary $U$ error remains close to that of the full method, whereas its $Z$-error and residual ratios reach $73.957$ and $6.462$. This pattern is consistent with the shared--local representation and replay constraint primarily supporting control learning across time levels rather than uniformly improving a value-only endpoint.

\medskip 
\noindent\textit{Finite-branch variance reduction.}
Disabling antithetic sampling or the control-label response baseline substantially increases the control error on Reaction--diffusion, with $Z$-error ratios of $64.948$ and $69.344$, respectively. On Quadratic-gradient, the same ablations slightly reduce the primary $U$ error while increasing both the $Z$ error and dynamic residual. This metric-dependent response is consistent with the intended role of these mechanisms: they act directly on the finite-branch Brownian labels, and their effect is therefore more visible in control accuracy and dynamic consistency than in every value-only metric.

\medskip 
\noindent\textit{Within-layer control refitting.}
Removing control refitting slightly increases the primary $U$ error on Quadratic-gradient while reducing the $Z$ error and residual to $0.909$ and $0.986$ of their full-configuration values, respectively. On Reaction--diffusion, the corresponding $U$-, $Z$-, and residual ratios are $1.658$, $1.444$, and $1.010$. Thus, control refitting can improve joint value--control performance.

% ============================================================================
\subsubsection{Sensitivity to the Structural-Interface Configuration}
\label{subsubsec:representation_results_revised}

The structural-interface study keeps the discrete backward backbone, generic training and coordination mechanisms, and prescribed budgets fixed while varying the candidate-source policies A--D. Table~\ref{tab:profiles_revised} defines the four policies, with D denoting the full \DCBSDE{} configuration and A--no-symmetry denoting a diagnostic variant of A in which both symmetry projections are set to the identity.

\begin{table}[H]
\centering
\caption{Candidate-source policies used by the common structural-screening interface.}
\label{tab:profiles_revised}
\small
\setlength{\tabcolsep}{4pt}
\renewcommand{\arraystretch}{1.08}
\begin{tabularx}{\textwidth}{@{}c
>{\raggedright\arraybackslash}X
>{\raggedright\arraybackslash}X
>{\raggedright\arraybackslash}X@{}}
\toprule
Configuration
& Symmetry candidate
& Value-feature candidates
& Response-baseline candidate \\
\midrule
A
& Automatically proposed and verified; identity if rejected
& Generic automatic feature library
& Automatic differentiation; zero if unavailable \\
B
& Same as A
& Screened equation-informed features
& Same as A \\
C
& Same as A
& Generic automatic feature library
& Validated equation-informed response; automatic-differentiation fallback \\
D
& Verified equation-declared candidate preferred; otherwise the rule of A
& Screened equation-informed features
& Validated equation-informed response; automatic-differentiation fallback \\
\bottomrule
\end{tabularx}
\end{table}

HJB-Quadratic is used to assess the effect of the symmetry projection, whereas Quadratic-gradient and Reaction--diffusion are used to examine sensitivity to the sources of the retained value features and response baselines. Table~\ref{tab:structural_results_revised} reports error ratios relative to the locked full \DCBSDE{} configuration. A ratio below $1$ indicates a smaller error for the candidate policy on the corresponding metric. The reported intervals are $95\%$ paired descriptive bootstrap intervals for the primary-$U$ error ratios.

\begin{table}[H]
\centering
\caption{Sensitivity to the structural-interface configuration.}
\label{tab:structural_results_revised}
\scriptsize
\setlength{\tabcolsep}{3.5pt}
\renewcommand{\arraystretch}{1.05}
\begin{tabularx}{\textwidth}{@{}lXcccc@{}}
\toprule
Benchmark
& Candidate policy
& \shortstack{Primary-$U$\\ratio}
& \shortstack{$95\%$ paired\\interval}
& \shortstack{$Z$\\ratio}
& \shortstack{Residual\\ratio} \\
\midrule
HJB-Quadratic
& A
& 1.000
& $[1.000,1.000]$
& 1.000
& 1.000 \\
HJB-Quadratic
& A--no-symmetry
& 10.267
& $[9.320,11.367]$
& 3.901
& 1.037 \\
\addlinespace[2pt]
Quadratic-gradient
& C
& 0.826
& $[0.814,0.837]$
& 1.288
& 2.439 \\
Quadratic-gradient
& B
& 1.000
& $[1.000,1.000]$
& 1.000
& 1.000 \\
Quadratic-gradient
& A
& 0.826
& $[0.815,0.837]$
& 1.288
& 2.439 \\
Quadratic-gradient
& A--no-symmetry
& 0.988
& $[0.975,1.001]$
& 2.253
& 2.522 \\
\addlinespace[2pt]
Reaction--diffusion
& C
& 1.115
& $[1.058,1.184]$
& 69.302
& 5.964 \\
Reaction--diffusion
& B
& 1.000
& $[1.000,1.000]$
& 1.000
& 1.000 \\
Reaction--diffusion
& A
& 1.115
& $[1.058,1.183]$
& 69.302
& 5.964 \\
\bottomrule
\end{tabularx}
\end{table}

\medskip 
\noindent\textit{Symmetry projection.}
For HJB-Quadratic, policy A induces the same effective configuration as the full policy, whereas forcing the symmetry projections to the identity increases the primary $U$ error by a factor of $10.267$ and the $Z$ error by a factor of $3.901$. The residual ratio increases only to $1.037$. Within this benchmark, the contrast is consistent with the verified symmetry projection reducing the effective hypothesis class by imposing the known invariance. The limited change in the residual also shows that a small one-step residual need not imply an accurate value or control approximation.

\medskip 
\noindent\textit{Value-feature and response-source sensitivity.}
For Quadratic-gradient, policies A and C induce the same effective retained configuration. Their primary $U$ error is $0.826$ of the full-configuration error, while their $Z$ and residual ratios increase to $1.288$ and $2.439$. Policy B induces the full configuration. The A--no-symmetry variant produces a primary $U$ ratio close to one, but its $Z$ and residual ratios rise to $2.253$ and $2.522$. Hence, a candidate-source policy that improves the function-value metric need not simultaneously improve control accuracy or dynamic
consistency.

For Reaction--diffusion, policies A and C again induce the same effective retained configuration, whereas policy B maps to the full configuration. Relative to the full method, A and C increase the primary $U$ error to $1.115$ and the $Z$ error to $69.302$. The much larger change in the control metric shows that the primary value error is comparatively insensitive to the deterioration in the learned control. This comparison is consistent with the alternative retained configuration providing a more suitable joint
value--control representation for this benchmark.

\FloatBarrier

% ============================================================================
\subsection{Computational Cost and Wall-Clock Comparison}
\label{subsec:wallclock_results_revised}

Both DBDP2 and \DCBSDE{} train backward one time level at a time, but they use different mechanisms to identify the control. In the wall-clock-matched experiments, DBDP2 is trained for approximately the same elapsed time as the complete \DCBSDE{} pipeline. This tests whether the accuracy gap can be explained merely by the shorter native training time of DBDP2.
Table~\ref{tab:wallclock_results} reports the resulting descriptive comparison. The error ratio is ``matched DBDP2/\DCBSDE{}.''

\begin{table}[H]
\centering
\caption{Wall-clock-matched comparison with DBDP2.}
\label{tab:wallclock_results}
\small
\resizebox{\textwidth}{!}{%
\begin{tabular}{@{}lrrrrrrr@{}}
\toprule
Benchmark & Matched budget (s) & \DCBSDE{} time (s) & DBDP2 time (s) & Time ratio
& \DCBSDE{} error & DBDP2 error [min,max] & Error ratio \\
\midrule
Quadratic-gradient & 2121.579 & 2121.6 & 2122.3 & 1.000
& 0.0373 & 0.1186 $[0.0979,0.1704]$ & 3.18 \\
Reaction--diffusion & 1411.963 & 1390.4 & 1412.7 & 1.016
& 0.00408 & 0.0286 $[0.0231,0.0422]$ & 7.01 \\
\bottomrule
\end{tabular}}
\end{table}

The matched time ratios are $1.000$ and $1.016$, whereas the corresponding DBDP2/\DCBSDE{} primary-error ratios are $3.18$ and $7.01$. Every paired seed favors \DCBSDE{} on both equations. Within these two configurations, merely extending the DBDP2 training time therefore does not remove the observed accuracy gap.

% ============================================================================
% ============================================================================
\section{Conclusion}
\label{sec:conclusion}

Motivated by Gaussian perturbation and conditional regression in denoising score matching, we proposed Deep-Control BSDE (\DCBSDE{}), a layerwise control-regression method for high-dimensional Markovian BSDEs. At each time level, the frozen successor value function defines an explicit Brownian conditional-moment target for the current control. The control is regressed first, the value is then updated through the implicit BSDE relation, and the two approximations are coordinated before being stored as immutable snapshots. The Gaussian--Stein identity further relates this control target to the diffusion-directional gradient of the successor value function.

The hierarchy $z_n^\pi$, $z_n^{\mathrm{tar}}$, $\widehat Z_n^M$, and $Z_n^\theta$ separates time-discretization, successor-propagation, finite-branch, and local learning errors. Antithetic Brownian pairing and a frozen linear-response baseline reduce finite-branch fluctuations without changing the conditional mean of the labels. Generic coordination mechanisms stabilize layerwise training, while optional equation-informed structural interfaces improve representation without altering the population control target or the backward dependency structure. We establish stability of the exact discrete operators, residual-driven backward stability of the learned recursion, and conditional consistency when the discretization, local learning and numerical, and finite-branch errors vanish jointly.

Across six numerical configurations, \DCBSDE{} achieves the lowest mean primary error in five cases. Among benchmarks with complete pathwise references, it simultaneously improves the $U$- and $Z$-path errors in three cases. Burgers--20 instead exhibits a clear Pareto trade-off between value and control accuracy. Structural comparisons and component ablations further show that a small value error does not necessarily imply an accurate or stable value gradient, highlighting the importance of jointly evaluating value accuracy, control accuracy, and dynamic consistency. Wall-clock-matched experiments also indicate that the observed gains cannot generally be reproduced by simply extending the training time of a baseline method.

Beyond estimating $u(0,x_0)$, \DCBSDE{} provides approximations $U_n(x)\approx u(t_n,x), \qquad
Z_n(x)\approx \sigma(t_n,x)^\top\nabla_xu(t_n,x)$ at multiple time levels. It is therefore particularly suitable for applications requiring reliable joint approximation of the value and control processes, gradient sensitivities, or feedback policies, especially when control errors significantly affect backward value propagation through the generator.

Future work will extend the analysis beyond the global Lipschitz setting, develop adaptive allocation of time steps, training states, Brownian branches, and coordination budgets, and evaluate whether improved control accuracy translates into better performance in stochastic control, risk management, derivative pricing and hedging, and high-dimensional feedback-policy learning.

%%%%%%%%%%%%%%%%%%%%%%%%%%%%%%%%%%%%%%%%%%%%%%%%%%%%%%%%%%%%%%%%%%%%%%%%%%%%%%%%%%%%%%%%%%%%%%%%%%%%%%%%%%%%%%%%%%%%%%%%%%%%%%%%%%%%%%%%%%%%%%%%%%%%%%%%%%%%%%%%%%%%%%%%%%%%%%%%%%%%%%%%%%%%%%%%%%%%%%%%%%%%%%%%%%%%%%%%%%%%%%%%%%%%

\appendix

\providecommand{\dd}{\mathrm{d}}
\providecommand{\one}{\bm{1}}
\providecommand{\sigmoid}{\operatorname{sigmoid}}
\providecommand{\DCBSDE}{\textsc{DC-BSDE}}
% ============================================================================
% ============================================================================
\section{Benchmark Equations, Reference Information, and Evaluation Endpoints}
\label{app:benchmark_definitions}

This section gives a unified description of the continuous models, temporal discretization parameters, terminal conditions, reference information, and evaluation endpoints used for the six main numerical configurations and the supplementary stress tests. Unless stated otherwise, all pathwise errors are computed along independent test trajectories.

Unless otherwise specified, every benchmark is written in the form
\begin{align*}
    \dd X_t&=\mu(t,X_t)\,\dd t+\sigma(t,X_t)\,\dd W_t, \qquad X_0=x_0,\\
    Y_t&=g(X_T)+\int_t^T f(s,X_s,Y_s,Z_s)\,\dd s-\int_t^T Z_s^{\mathsf T}\,\dd W_s.
\end{align*}
We adopt the Markovian convention
\[
Y_t=u(t,X_t), \qquad
Z_t=\sigma(t,X_t)^{\mathsf T}\nabla_xu(t,X_t),
\]
so that the corresponding semilinear parabolic PDE is
\[
\begin{cases}
\displaystyle
    \partial_tu+\mu(t,x)^{\mathsf T}\nabla_xu
    +\dfrac12\operatorname{tr}\!\left(
       \sigma(t,x)\sigma(t,x)^{\mathsf T}D_x^2u
     \right)
    +f\!\left(t,x,u,\sigma(t,x)^{\mathsf T}\nabla_xu\right)=0,\\[1.5mm]
    u(T,x)=g(x).
\end{cases}
\]
The forward processes have constant diffusion coefficients, and all benchmarks use Euler--Maruyama time discretization.

The appendix contains both benchmarks satisfying the baseline Lipschitz assumptions of Section~\ref{sec:preliminaries} and empirical extensions beyond that setting. HJB-Quadratic, Paper-HJB, and Quadratic-gradient contain quadratic-gradient terms. Allen--Cahn contains a cubic reaction term that is not globally Lipschitz. The Burgers equation contains a multiplicative value--control coupling. These examples are used to test finite-budget robustness beyond the scope of the baseline theory.

% ============================================================================
\subsection{HJB-Quadratic}
\label{appsubsec:hjb_quadratic_definition}

\[
    d=50,
    \qquad T=1,
    \qquad N=40,
    \qquad x_0=\bm0,
\]
and
\[
\dd X_t=\dd W_t,\qquad f(t,x,y,z)=-\frac12\lVert z\rVert_2^2, \qquad g(x)=\lVert x\rVert_2^2.
\]
The corresponding PDE is
\[
\partial_tu+\frac12\Delta u-\frac12\lVert\nabla_xu\rVert_2^2=0,\qquad u(T,x)=\lVert x\rVert_2^2.
\]
The analytic solution and control are
\begin{align*}
    u(t,x)&=\frac d2\log\!\left(1+2(T-t)\right)+\frac{\lVert x\rVert_2^2}{1+2(T-t)},\\
    z(t,x)&=\frac{2x}{1+2(T-t)}.
\end{align*}
The primary metric is the $U$-path RMSE, and the $Z$-path RMSE is also reported. The initial value is $u(0,\bm0)=25\log 3\approx27.4653$. The solution depends only on $\lVert x\rVert_2^2$ and therefore has an explicit orthogonal-invariance structure.

% ============================================================================
\subsection{Cole--Hopf HJB (Paper-HJB)}
\label{appsubsec:paper_hjb_definition}

To distinguish it from HJB-Quadratic, we refer to the 100-dimensional HJB benchmark from the literature as Cole--Hopf HJB; its experiment-archive label is Paper-HJB. The parameters are
\[
    d=100,
    \qquad T=1,
    \qquad N=20,
    \qquad x_0=\bm0,
\]
with
\[
\dd X_t=\sqrt2\,\dd W_t, \qquad
f(t,x,y,z)=-\frac12\lVert z\rVert_2^2,\qquad
g(x)=\log\!\left(\frac{1+\lVert x\rVert_2^2}{2}\right).
\]
The corresponding PDE is
\[
\partial_tu+\Delta u-\lVert\nabla_xu\rVert_2^2=0, \qquad
u(T,x)=\log\!\left(\frac{1+\lVert x\rVert_2^2}{2}\right).
\]
The Cole--Hopf transformation gives the probabilistic representation
\[
u(t,x)=-\log\mathbb{E}\!\left[ \exp\!\left(-g\bigl(x+\sqrt2 W_{T-t}\bigr)\right) \right].
\]
The experiments use the independent Monte Carlo reference
\[
 u_{\mathrm{ref}}(0,x_0)=4.5901
\]
and take $|\widehat u(0,x_0)-4.5901|$ as the primary metric \cite{EHanJentzen2017}. This initial-value reference is not used to construct a full-path reference without independent error control.

% ============================================================================
\subsection{Allen--Cahn}
\label{appsubsec:allen_cahn_definition}

This benchmark follows the standard 100-dimensional Allen--Cahn setting \cite{EHanJentzen2017}:
\[
    d=100,
    \qquad T=0.3,
    \qquad N=20,
    \qquad x_0=\bm0,
\]
with
\[
\dd X_t=\sqrt2\,\dd W_t,\qquad
 f(t,x,y,z)=y-y^3, \qquad
    g(x)=\frac{1}{2+\frac25\lVert x\rVert_2^2}.
\]
The corresponding PDE is
\[
\partial_tu+\Delta u+u-u^3=0,\qquad
u(T,x)=\frac{1}{2+\frac25\lVert x\rVert_2^2},
\]
and $Z_t=\sqrt2\,\nabla_xu(t,X_t)$. The experiments use the independent branching-diffusion reference
\[
u_{\mathrm{ref}}(0,x_0)=0.052802
\]
and take $|\widehat u(0,x_0)-0.052802|$ as the primary metric.

% ============================================================================
\subsection{Burgers--20}
\label{appsubsec:burgers_definitions}

For a given dimension $d$, let
\[
\dd X_t=d\,\dd W_t,\qquad
\sigma=dI_d,
\]
and
\begin{align*}
    f(t,x,y,z) &=\left(y-\frac{d+2}{2d}\right)\sum_{i=1}^d z_i,\\
    g(x) &=\sigmoid\left(T+\frac1d\sum_{i=1}^d x_i\right).
\end{align*}
The corresponding PDE is
\[
\partial_tu+\frac{d^2}{2}\Delta u +d\left(u-\frac{d+2}{2d}\right) \sum_{i=1}^d\partial_{x_i}u=0.
\]
The analytic solution and control are
\begin{align*}
    u(t,x) &= \ sigmoid \left(t+\frac1d\sum_{i=1}^d x_i\right),\\
    z(t,x) &=u(t,x)\bigl[1-u(t,x)\bigr]\one_d.
\end{align*}
Burgers--20 uses
\[
    d=20,
    \qquad T=1,
    \qquad N=80,
    \qquad x_0=\bm0,
\]
and is a main-text configuration with $U$-path RMSE as the primary metric,  $u(0,\bm0)=1/2$. The solution depends only on the coordinate mean and therefore has a low-dimensional ridge structure.

% ============================================================================
\subsection{Quadratic-Gradient}
\label{appsubsec:quadratic_gradient_definition}

\[
    d=100,
    \qquad T=1,
    \qquad N=30,
    \qquad x_0=\bm0,
    \qquad \alpha=0.4,
\]
and define
\begin{align*}
    s(t,x)&=T-t+\lVert x\rVert_2^2,
    &F(s)&=\sin(s^\alpha),
    &\psi_{\mathrm{QG}}(t,x)&=F(s(t,x)).
\end{align*}
The required derivatives are
\begin{align*}
    F'(s) &=\alpha s^{\alpha-1}\cos(s^\alpha),\\
    F''(s) &=\alpha(\alpha-1)s^{\alpha-2}\cos(s^\alpha)-\alpha^2s^{2\alpha-2}\sin(s^\alpha),\\
    \partial_t\psi_{\mathrm{QG}}&=-F'(s),\\
    \nabla_x\psi_{\mathrm{QG}} &=2F'(s)x,\\
    \Delta\psi_{\mathrm{QG}} &=2dF'(s)+4\lVert x\rVert_2^2F''(s).
\end{align*}
The forward process is $\dd X_t=\dd W_t$. Define
\[
f(t,x,y,z)={}\lVert z\rVert_2^2-\lVert\nabla_x\psi_{\mathrm{QG}}(t,x)\rVert_2^2
      -\partial_t\psi_{\mathrm{QG}}(t,x)-\frac12\Delta\psi_{\mathrm{QG}}(t,x),
\]

and
\[
g(x)=\sin\!\left(\lVert x\rVert_2^{2\alpha}\right).
\]
Then
\[
\begin{aligned}
0={}&\partial_t u+\frac12\Delta u+\lVert\nabla_x u\rVert_2^2-\lVert\nabla_x\psi_{\mathrm{QG}}\rVert_2^2 -\partial_t\psi_{\mathrm{QG}}-\frac12\Delta\psi_{\mathrm{QG}}.
\end{aligned}
\]
and the manufactured exact solution and control are
\begin{align*}
    u(t,x)&=\psi_{\mathrm{QG}}(t,x)=\sin\!\left[ \left(T-t+\lVert x\rVert_2^2\right)^\alpha\right],\\
    z(t,x)&=2F'(s(t,x))x.
\end{align*}
The primary metric is the $U$-path RMSE, and the $Z$-path RMSE is also reported. The initial value is $u(0,\bm0)=\sin(1)\approx0.8414709848$. Since $\alpha<1$, the negative powers in the derivatives can be numerically unstable as $s\downarrow0$. The implementation evaluates the relevant derivatives using the safeguard $s_\varepsilon=\max\{s,10^{-12}\}$, while retaining the original terminal function and analytic function values for evaluation.

% ============================================================================
\subsection{Reaction--Diffusion}
\label{appsubsec:reaction_diffusion_definition}

This manufactured-solution benchmark uses
\[
    d=100,
    \qquad T=1,
    \qquad N=30,
    \qquad x_0=\bm0,
    \qquad \kappa=0.6,
    \qquad \lambda=d^{-1/2},
\]
and
\[
\psi_{\mathrm{RD}}(t,x)=1+\kappa +\sin\!\left(\lambda\sum_{i=1}^d x_i\right)\exp\!\left[\frac{\lambda^2d}{2}(t-T)\right].
\]
The forward process is $\dd X_t=\dd W_t$, and the generator and terminal condition are
\[
\begin{aligned}
    f(t,x,y,z) &=\min\left\{1,\left[y-\psi_{\mathrm{RD}}(t,x)\right]^2 \right\},\\
    g(x) &=1+\kappa+\sin\!\left(\lambda\sum_{i=1}^d x_i\right).
\end{aligned}
\]
The corresponding PDE is
\[
\partial_tu+\frac12\Delta u+\min\left\{1, \left[u-\psi_{\mathrm{RD}}(t,x)\right]^2 \right\}=0.
\]
Its exact solution and control are
\begin{align*}
    u(t,x)&=\psi_{\mathrm{RD}}(t,x),\\
    z(t,x)&=\lambda\cos\!\left(\lambda\sum_{i=1}^d x_i\right)\exp\!\left[\frac{\lambda^2d}{2}(t-T)\right]\one_d.
\end{align*}
The primary metric is the $U$-path RMSE, and the $Z$-path RMSE is also reported. The initial value is $u(0,\bm0)=1+\kappa=1.6$.

% ============================================================================
\section{Complete Results for the Main Experiments}
\label{app:extended_experiments}
% ============================================================================
\subsection{Primary-Error Comparisons with Classical Baselines and DPI}
\label{appsubsec:all_pairwise}

Table~\ref{tab:all_pairwise_main} reports the comparisons with DeepBSDE-style, DBDP1, and DBDP2. These methods use the original main-matrix primary metrics and the common paired-seed protocol, with ten paired runs for every benchmark--method combination. 

Table~\ref{tab:dpi_summary} combines the DPI primary-error comparisons with the available DPI control-path diagnostics. The primary-error ratios use the same benchmark-specific absolute primary metrics as in the main text. For HJB-Quadratic and Paper-HJB, the DPI numerator is the conditional mean over the $5$ and $1$ runs passing the stated stability audit, respectively, whereas the remaining four DPI comparisons use all ten runs. The \DCBSDE{} denominator always uses all ten seeds. The HJB-Quadratic and Paper-HJB ratios are therefore audit-conditioned descriptive summaries.

\begin{table}[H]
\centering
\caption{Classical baseline--\DCBSDE{} comparisons under the common paired-seed protocol}
\label{tab:all_pairwise_main}
\scriptsize
\resizebox{\textwidth}{!}{%
\begin{tabular}{@{}llcc@{}}
\toprule
Benchmark
& Baseline
& Baseline/\DCBSDE{}
& \shortstack{$95\%$ paired descriptive\\bootstrap interval} \\
\midrule
HJB-Quadratic & DeepBSDE-style
& 19.753
& $[15.763,24.303]$ \\
HJB-Quadratic & DBDP1
& 172.468
& $[153.551,194.144]$ \\
HJB-Quadratic & DBDP2
& 173.218
& $[154.365,195.759]$ \\
\midrule
Allen--Cahn & DeepBSDE-style
& 2.943
& $[1.492,5.072]$ \\
Allen--Cahn & DBDP1
& 2.325
& $[1.218,3.480]$ \\
Allen--Cahn & DBDP2
& 1.551
& $[1.001,2.166]$ \\
\midrule
Burgers--20 & DeepBSDE-style
& 0.844
& $[0.709,0.975]$ \\
Burgers--20 & DBDP1
& 6.122
& $[5.292,7.037]$ \\
Burgers--20 & DBDP2
& 5.638
& $[4.919,6.479]$ \\
\midrule
Paper-HJB & DeepBSDE-style
& 12.347
& $[8.596,17.565]$ \\
Paper-HJB & DBDP1
& 5289.680
& $[4399.983,6639.882]$ \\
Paper-HJB & DBDP2
& 13.533
& $[9.132,19.478]$ \\
\midrule
Quadratic-gradient & DeepBSDE-style
& 5.011
& $[4.945,5.079]$ \\
Quadratic-gradient & DBDP1
& 1.967
& $[1.896,2.047]$ \\
Quadratic-gradient & DBDP2
& 1.864
& $[1.802,1.934]$ \\
\midrule
Reaction--diffusion & DeepBSDE-style
& 3.164
& $[2.830,3.538]$ \\
Reaction--diffusion & DBDP1
& 16.332
& $[15.000,17.912]$ \\
Reaction--diffusion & DBDP2
& 14.890
& $[13.630,16.373]$ \\
\bottomrule
\end{tabular}}
\end{table}

\begin{table}[H]
\centering
\caption{DPI primary-error comparisons and available control-path diagnostics}
\label{tab:dpi_summary}
\small
\begin{tabular}{@{}lccc@{}}
\toprule
Benchmark
& Runs used/planned
& \shortstack{DPI/\DCBSDE{}\\primary-error ratio}
& \shortstack{DPI $Z$-path RMSE:\\mean $\pm$ sample SD} \\
\midrule
HJB-Quadratic
& $5/10$
& 102.217
& $0.5622\pm0.0663$ \\
Allen--Cahn
& $10/10$
& 6.869
& N/A \\
Burgers--20
& $10/10$
& 0.275
& $0.0021\pm1.440\times10^{-4}$ \\
Paper-HJB
& $1/10$
& 1215.209
& N/A \\
Quadratic-gradient
& $10/10$
& 2.508
& $0.0322\pm0.0031$ \\
Reaction--diffusion
& $10/10$
& 3.135
& $0.0031\pm3.972\times10^{-4}$ \\
\bottomrule
\end{tabular}
\end{table}

\subsection{Available Secondary Diagnostics for the Six Configurations}
\label{appsubsec:secondary_diagnostics}

Tables~\ref{tab:secondary_diag_part1}--\ref{tab:secondary_diag_part2} report the secondary diagnostics for \DCBSDE{}, DeepBSDE-style, DBDP1, and DBDP2. The three metrics quantify, respectively, control-path accuracy, one-step pathwise innovation consistency, and accumulated terminal consistency. All entries for these four methods are reported as the group mean followed by the minimum--maximum range over the ten seeds. In the $Z$-path column, N/A indicates that no validated reference control path is available for the corresponding benchmark.

\begin{table}[H]
\centering
\caption{Secondary diagnostics for HJB-Quadratic, Allen--Cahn, and Burgers--20}
\label{tab:secondary_diag_part1}
\scriptsize
\resizebox{\textwidth}{!}{%
\begin{tabular}{@{}llccc@{}}
\toprule
Benchmark
& Method
& $Z$-path RMSE: mean [min,max]
& BSDE residual RMSE: mean [min,max]
& Terminal rollout RMSE: mean [min,max] \\
\midrule
HJB-Quadratic & \DCBSDE{}
& $0.0768\ [0.0687,0.0847]$
& $0.1631\ [0.1628,0.1635]$
& $0.9686\ [0.9636,0.9713]$ \\
HJB-Quadratic & DeepBSDE-style
& $0.7567\ [0.7082,0.8463]$
& $1.3471\ [1.1040,1.8609]$
& $1.5729\ [1.4770,1.7145]$ \\
HJB-Quadratic & DBDP1
& $5.0045\ [5.0039,5.0055]$
& $22.490\ [22.412,22.561]$
& $50.357\ [50.330,50.379]$ \\
HJB-Quadratic & DBDP2
& $5.0058\ [5.0057,5.0058]$
& $5.5226\ [5.4244,5.6416]$
& $12.403\ [12.180,12.601]$ \\
\midrule
Allen--Cahn & \DCBSDE{}
& N/A
& $3.881\times10^{-4}\ [3.870\times10^{-4},3.896\times10^{-4}]$
& $0.00152\ [0.00152,0.00153]$ \\
Allen--Cahn & DeepBSDE-style
& N/A
& $0.00185\ [0.00149,0.00248]$
& $0.00254\ [0.00208,0.00271]$ \\
Allen--Cahn & DBDP1
& N/A
& $0.00242\ [0.00240,0.00244]$
& $0.00562\ [0.00558,0.00568]$ \\
Allen--Cahn & DBDP2
& N/A
& $0.00235\ [0.00232,0.00238]$
& $0.00520\ [0.00517,0.00523]$ \\
\midrule
Burgers--20 & \DCBSDE{}
& $0.0138\ [0.0114,0.0179]$
& $0.00958\ [0.00948,0.00973]$
& $0.0481\ [0.0474,0.0495]$ \\
Burgers--20 & DeepBSDE-style
& $0.0317\ [0.0299,0.0339]$
& $0.0108\ [0.0101,0.0113]$
& $0.0472\ [0.0464,0.0482]$ \\
Burgers--20 & DBDP1
& $0.1547\ [0.1402,0.1735]$
& $0.0307\ [0.0290,0.0325]$
& $0.1001\ [0.0946,0.1165]$ \\
Burgers--20 & DBDP2
& $0.1530\ [0.1422,0.1639]$
& $0.0248\ [0.0239,0.0265]$
& $0.1023\ [0.0953,0.1143]$ \\
\bottomrule
\end{tabular}}
\end{table}

\begin{table}[H]
\centering
\caption{Secondary diagnostics for Paper-HJB, Quadratic-gradient, and Reaction--diffusion}
\label{tab:secondary_diag_part2}
\scriptsize
\resizebox{\textwidth}{!}{%
\begin{tabular}{@{}llccc@{}}
\toprule
Benchmark
& Method
& $Z$-path RMSE: mean [min,max]
& BSDE residual RMSE: mean [min,max]
& Terminal rollout RMSE: mean [min,max] \\
\midrule
Paper-HJB & \DCBSDE{}
& N/A
& $0.00757\ [0.00756,0.00758]$
& $0.0337\ [0.0336,0.0338]$ \\
Paper-HJB & DeepBSDE-style
& N/A
& $0.0381\ [0.0231,0.0656]$
& $0.0518\ [0.0504,0.0539]$ \\
Paper-HJB & DBDP1
& N/A
& $1.7112\ [1.7031,1.7208]$
& $3.8293\ [3.8103,3.8520]$ \\
Paper-HJB & DBDP2
& N/A
& $0.0639\ [0.0631,0.0655]$
& $0.1504\ [0.1495,0.1509]$ \\
\midrule
Quadratic-gradient & \DCBSDE{}
& $0.0660\ [0.0646,0.0675]$
& $0.0161\ [0.0161,0.0162]$
& $0.1168\ [0.1164,0.1174]$ \\
Quadratic-gradient & DeepBSDE-style
& $0.1552\ [0.1542,0.1565]$
& $0.0865\ [0.0806,0.1003]$
& $0.3761\ [0.3751,0.3777]$ \\
Quadratic-gradient & DBDP1
& $0.1581\ [0.1556,0.1621]$
& $0.0965\ [0.0948,0.1000]$
& $0.3724\ [0.3701,0.3739]$ \\
Quadratic-gradient & DBDP2
& $0.1568\ [0.1543,0.1592]$
& $0.0947\ [0.0909,0.0969]$
& $0.3756\ [0.3733,0.3789]$ \\
\midrule
Reaction--diffusion & \DCBSDE{}
& $0.00651\ [0.00562,0.00724]$
& $0.0112\ [0.0111,0.0113]$
& $0.0585\ [0.0584,0.0587]$ \\
Reaction--diffusion & DeepBSDE-style
& $0.0211\ [0.0192,0.0226]$
& $0.0141\ [0.0135,0.0153]$
& $0.0622\ [0.0614,0.0628]$ \\
Reaction--diffusion & DBDP1
& $0.1334\ [0.1254,0.1402]$
& $0.4808\ [0.4753,0.4848]$
& $2.0563\ [2.0551,2.0577]$ \\
Reaction--diffusion & DBDP2
& $0.1249\ [0.1169,0.1387]$
& $0.0468\ [0.0459,0.0482]$
& $0.1749\ [0.1699,0.1806]$ \\
\bottomrule
\end{tabular}}
\end{table}

% ============================================================================
% ============================================================================
\section{LogHeat: A Verifiable Intersection between the Score and the BSDE Control}
\label{app:score_case}

For a general BSDE, the control process is not the score of a probability density. The LogHeat problem nevertheless provides a verifiable special case in which the two objects coincide after multiplication by the diffusion matrix. Let the positive function $\rho$ solve the backward heat equation generated by the diffusion matrix $\sigma$:
\begin{equation}
\partial_t\rho(t,x)+\frac12\operatorname{tr}\!\left(a\nabla_x^2\rho(t,x)\right)=0,\qquad
\rho(T,x)=\rho_T(x)>0,\qquad a=\sigma\sigma^\top.
\label{eq:logheat_density_revised}
\end{equation}
Set $u=\log\rho$. Then
\[
\partial_tu+\frac12\operatorname{tr}\!\left(a\nabla_x^2u\right)+\frac12\left\lVert\sigma^{\mathsf T}\nabla_xu\right\rVert_2^2=0,\qquad
u(T,x)=\log\rho_T(x).
\]
For the forward process $\dd X_t=\sigma\,\dd W_t$, the associated BSDE has generator $f(z)=\tfrac12\lVert z\rVert_2^2$, and
\begin{equation}
Z_t=\sigma^\top\nabla_xu(t,X_t)=\sigma^\top\nabla_x\log\rho(t,X_t).
\label{eq:logheat_score_control_revised}
\end{equation}
Thus, if $s(t,x)=\nabla_x\log\rho(t,x)$ denotes the density score, then $Z=\sigma^{\mathsf T}s$. This identity is specific to the logarithmically transformed class described by equations~\eqref{eq:logheat_density_revised}--\eqref{eq:logheat_score_control_revised}. It does not imply that the control process of a general BSDE is the score of a state-marginal density.

The experiment uses a positive Gaussian-mixture terminal density and considers only the setting $d=100$ and $N=20$. Control error is the prespecified evaluation endpoint. Each of the four BSDE methods uses $n=5$ paired random seeds. The main tabulated results are evaluated along paired BSDE path states. The time-layer plot additionally compares the control error on states sampled from the score density.

Table~\ref{tab:logheat_method_results} reports the relative $L^2$ error of the $Z$-path and the relative $Z$-error on the primary evaluation states.

\begin{table}[H]
\centering
\caption{Control errors for the $d=100$ LogHeat problem. }
\label{tab:logheat_method_results}
\small
\resizebox{\textwidth}{!}{%
\begin{tabular}{@{}lrrrr@{}}
\toprule
Method & $n$ & Relative $L^2$ error of the $Z$-path & $95\%$ descriptive interval & Relative $Z$-error on primary evaluation states \\
\midrule
DBDP1 & 5 & 1.9137 & $[1.9057,1.9218]$ & 2.2034 \\
DBDP2 & 5 & 2.7607 & $[2.4009,3.1204]$ & 2.4865 \\
DeepBSDE-style & 5 & 0.3681 & $[0.2277,0.5084]$ & 0.3324 \\
\DCBSDE{} & 5 & $\mathbf{0.0601}$ & $\mathbf{[0.0586,0.0615]}$ & $\mathbf{0.0714}$ \\
\bottomrule
\end{tabular}}
\end{table}

The relative $L^2$ error of the $Z$-path for \DCBSDE{} is $0.0601$, below those of all three BSDE baselines. The corresponding errors for DeepBSDE-style, DBDP1, and DBDP2 are $0.3681$, $1.9137$, and $2.7607$, respectively. In this special case, where the control--score relation is analytically verifiable, the results indicate that layerwise Brownian control supervision recovers the control path accurately.

Figure~\ref{fig:logheat_time_layers} shows that the DSM score baseline maintains a low error under direct score supervision. Among the four BSDE solvers, \DCBSDE{} exhibits the lowest and most stable control error. DeepBSDE-style ranks second, whereas DBDP1 and DBDP2 show larger biases or greater cross-seed variability at several time layers. The comparison demonstrates that direct score learning provides a strong auxiliary baseline when a genuine density-score structure is available, while \DCBSDE{} remains capable of recovering the corresponding control within a general BSDE-solver framework.

\begin{figure}[H]
\centering
\includegraphics[width=0.94\textwidth]{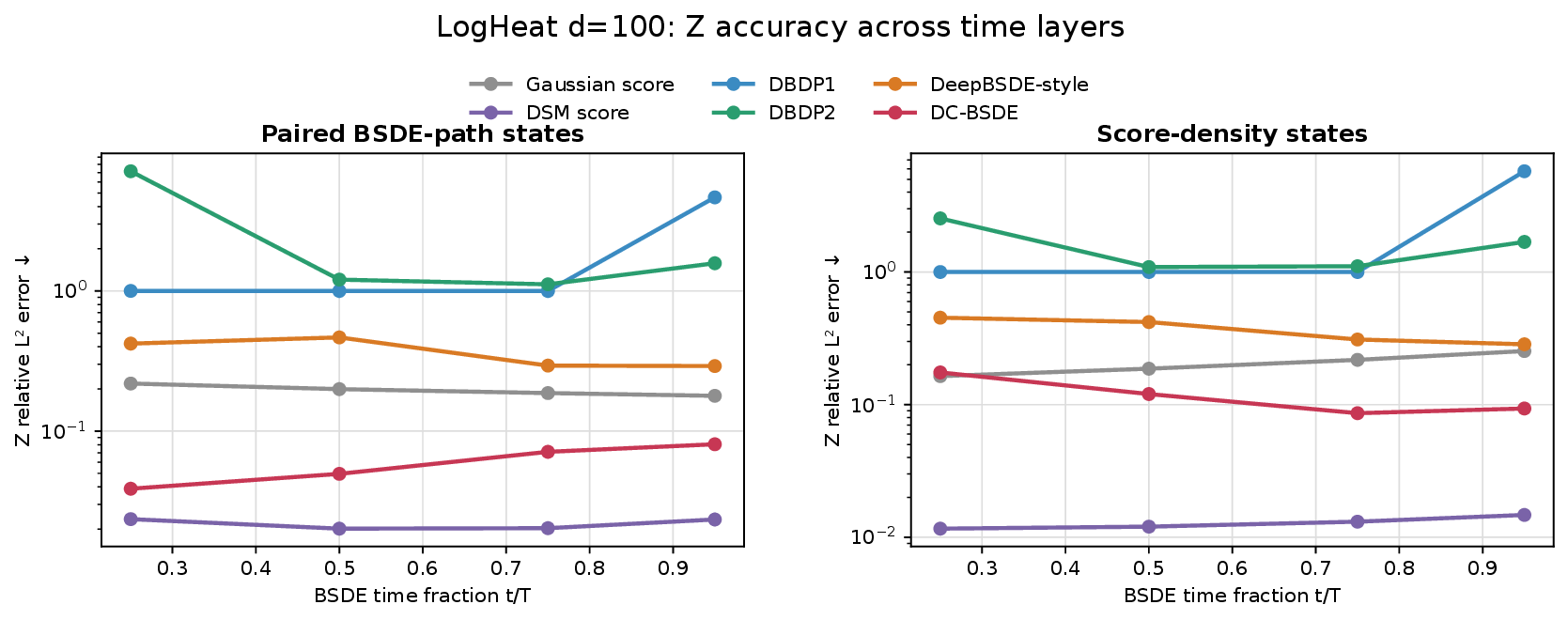}
\caption{Relative $L^2$ error of $Z$ across time layers for the $d=100$ LogHeat problem}
\label{fig:logheat_time_layers}
\end{figure}

% ============================================================================
% ============================================================================

%\clearpage
\FloatBarrier

\end{document}